\documentclass[a4paper,USenglish,thm-restate,nolineno]{socg-lipics-v2021}
\usepackage[utf8]{inputenc}
\usepackage{mathtools}
\usepackage{thmtools} 

\newtheorem{lem}[theorem]{Lemma}
\usepackage{xspace}
\usepackage{enumitem} 
\DeclarePairedDelimiter\floor{\lfloor}{\rfloor}

\usepackage{url}
\usepackage[normalem]{ulem}

\usepackage{makecell}

\usepackage{graphicx}
\graphicspath{{./figs/}}

\usepackage{todonotes}

\title{Twisted Ways to Find Plane Structures in Simple Drawings of Complete Graphs}

\author{Oswin Aichholzer}{Institute of Software Technology, Graz University of Technology, Austria
}{oaich@ist.tugraz.at}{https://orcid.org/0000-0002-2364-0583}{Partially supported by the Austrian Science Fund (FWF): W1230 and by H2020-MSCA-RISE project 734922 - CONNECT.}
\author{Alfredo Garc\'ia}{Departamento de M\'etodos Estad\'isticos and IUMA, Universidad de Zaragoza, Spain}{olaverri@unizar.es}{http://orcid.org/0000-0002-6519-1472}{Supported by H2020-MSCA-RISE project 734922 - CONNECT and Gobierno de Aragón project E41-17R.}
\author{Javier Tejel}{Departamento de M\'etodos Estad\'isticos and IUMA, Universidad de Zaragoza, Spain}{jtejel@unizar.es}{https://orcid.org/0000-0002-9543-7170}{Supported by H2020-MSCA-RISE project 734922 - CONNECT,
	Gobierno de Aragón project E41-17R and project PID2019-104129GB-I00 /
	AEI / 10.13039/501100011033 of the Spanish Ministry of Science and
	Innovation.}
\author{Birgit Vogtenhuber}{Institute of Software Technology, Graz University of Technology, Austria
}{bvogt@ist.tugraz.at}{https://orcid.org/0000-0002-7166-4467}{Partially supported by Austrian Science Fund within the collaborative DACH project \emph{Arrangements and Drawings} as FWF project \mbox{I 3340-N35} and by H2020-MSCA-RISE project 734922 - CONNECT.}
\author{Alexandra Weinberger}{Institute of Software Technology, Graz University of Technology, Austria}{weinberger@ist.tugraz.at
}{https://orcid.org/0000-0001-8553-6661}{Supported by the Austrian Science Fund (FWF): W1230 and by H2020-MSCA-RISE project 734922 - CONNECT.}

\authorrunning{O.~Aichholzer,  A.~Garc\'ia, J.~Tejel, B.~Vogtenhuber, and A.~Weinberger}

\Copyright{{Oswin Aichholzer and Alfredo Garc\'ia and Javier Tejel and Birgit~Vogtenhuber and Alexandra~Weinberger}} 



\pdfoutput=1 
\hideLIPIcs  

\newcommand{\cmonotone}{c-monotone}
\newcommand{\cmonotonicity}{c-monotonicity}

\newcommand{\gtwisted}{generalized twisted}

\newcommand{\Gtwisted}{Generalized twisted}

\newcommand\blfootnote[1]{%
	\begingroup
	\renewcommand\thefootnote{}\footnote{#1}%
	\addtocounter{footnote}{-1}%
	\endgroup
}

\begin{document}
	
	\maketitle
	\begin{abstract}
		Simple drawings are drawings of graphs in which the edges are Jordan arcs and each pair of edges share at most one point (a proper crossing or a common endpoint). 
		We introduce a special kind of simple drawings that we call generalized twisted drawings. 
		A simple drawing is generalized twisted if there is a point~$O$ such that every ray emanating from~$O$ crosses every edge of the drawing at most once and there is a ray emanating from~$O$ which crosses every edge exactly once.
		
		Via this new class of simple drawings, we show that every simple drawing of the complete graph with~$n$ vertices contains~$\Omega(n^{\frac{1}{2}})$ pairwise disjoint edges and a plane path of length~$\Omega(\frac{\log n }{\log \log n})$.
		Both results improve over previously known best lower bounds.
		On the way we show several structural results about and properties of generalized twisted drawings.
		We further present different characterizations of generalized twisted drawings, which might be of independent interest.
		\blfootnote{\hspace{-0.3cm}\begin{minipage}[l]{0.99\textwidth} \it 
			This work (without appendix) is available at the 38th International Symposium on Computational Geometry (SoCG 2022).
			Some results of this work have also been presented at the Computational Geometry:~Young Researchers Forum in 2021~\cite{YRF} and at the Encuentros de Geometr\'{\i}a Computacional 2021~\cite{EGC}.\end{minipage}}
	\end{abstract}

	\section{Introduction}\label{sec:intro}
	
	Simple drawings are drawings of graphs in the plane such that vertices are distinct points in the plane, edges are Jordan arcs connecting their endpoints, and edges intersect at most once either in a proper crossing or in a shared endpoint.
	The edges and vertices of a drawing partition the plane (or, more exactly, the plane minus the drawing) into regions,
	which are called the \emph{cells} of the drawing. 
	If a simple drawing is plane (that is, crossing-free), then its cells are classically called \emph{faces}.

	In the past decades, there has been significant interest in simple drawings. Questions about plane subdrawings of simple drawings of the complete graph on $n$ vertices, $K_n$, have attracted particularly close attention.
	
	Rafla~\cite{rafla} conjectured that every simple drawing of $K_n$ contains a plane Hamiltonian cycle. The conjecture has been shown to hold for~$n \leq 9$~\cite{all_small_drawings}, as well as for several special classes of simple drawings, like straight-line, monotone, and cylindrical drawings, but remains open in general.
	If Rafla's conjecture is true, then this would immediately imply that every simple drawing of the complete graph contains a plane perfect matching. However, to-date even the existence of such a matchging is still unknown.
	
	Ruiz-Vargas~\cite{RUIZVARGAS2017} showed in 2017 that every simple drawing of 
	$K_n$ contains 
	$\Omega(n^{\frac{1}{2}-\varepsilon})$ pairwise disjoint edges for any $\varepsilon>0$, which improved over a series of previous results:
	$\Omega((\log n )^{\frac{1}{6}})$ in 2003~\cite{bound_2003}, 
	$\Omega(\frac{\log n }{\log \log n})$ in 2005~\cite{bound_2005}, 
	$\Omega((\log n)^{1+\varepsilon})$ in 2009~\cite{bound_2009}, 
	and $\Omega(n^{\frac{1}{3}})$ in 2013 and 2014~\cite{bound_2014, triangles_together, bound_2013}.
	
	Pach, Solymosi, and T{\'o}th~\cite{bound_2003} showed that every simple drawing of $K_n$ contains a subdrawing of $K_{c\log^{\frac{1}{8}}n}$, for some constant $c$, that is either \emph{convex} or \emph{twisted}\footnote{In their definition for simple drawings, \emph{convex} means that there is a labeling of the vertices to $v_1,v_2,...,v_n$ such that $(v_i,v_j)$ ($i<j$) crosses $(v_k,v_l)$ ($k<l$) if and only if $i<k<j<l$ or $k<i<l<j$, and \emph{twisted} means that there is a labeling of the vertices to $v_1,v_2,...,v_n$ such that $(v_i,v_j)$ ($i<j$) crosses $(v_k,v_l)$ ($k<l$) if and only if $i<k<l<j$ or $k<i<j<l$.}. They further showed that every simple drawing of $K_n$ contains a plane subdrawing isomorphic to any fixed tree with up to $c \log^{\frac{1}{6}}n$ vertices, for some constant $c$. This implies that every simple drawing of $K_n$ contains a plane path of length $\Omega((\log n )^{\frac{1}{6}})$, which has been the best lower bound known prior to this paper.
	
	Concerning general plane substructures, it follows from a result of Ruiz-Vargas~\cite{RUIZVARGAS2017} that every simple drawing of $K_n$ contains a plane subdrawing with at least $2n-3$ edges. Further, Garc\'ia, Pilz, and Tejel~\cite{biconnected} showed that every maximal plane subdrawing of a simple drawing of $K_n$ is biconnected. 
	Note that, in contrast to straight-line drawings, simple drawings of~$K_n$ in general do not contain triangulations, that is, plane subdrawings where all faces (except at most one) are 3-cycles.
	
	In this paper, we introduce a new family of simple drawings, which we call \emph{generalized twisted} drawings. The name stems from the fact that 
	one can show that any twisted drawing is weakly isomorphic 
	to a generalized twisted drawing
	 (but not every generalized twisted drawing is weakly  isomorphic 
	to a twisted drawing). 
	It follows, that for any $n$ there exists a {\gtwisted} drawing.
	Two drawings~$D$ and~$D'$ are \emph{weakly isomorphic} if there is a bijection between the vertices and edges of $D$ and $D'$ such that a pair of edges in $D$ crosses exactly when the corresponding pair of edges in $D'$ crosses.

	\begin{definition}
		A simple drawing~$D$ is \textbf{{\cmonotone}}
		(short for circularly monotone) if there is a point~$O$ such that any ray emanating from~$O$ intersects any edge of~$D$ at most once.
		
		A simple drawing~$D$ of~$K_n$ is \textbf{{\gtwisted}} if there is a point~$O$ such that~$D$ is {\cmonotone} with respect to $O$ and there exists a ray~$r$ emanating from~$O$ that intersects every edge of~$D$.
	\end{definition}
	
	We label the vertices of {\cmonotone} drawings $v_1, \ldots, v_n$ in counterclockwise order around~$O$.
	In {\gtwisted} drawings, they are labeled such that the ray $r$ emerges from~$O$ between the ray to $v_1$ and the one to $v_n$. 
	Figure~\ref{fig:example_gtwisted} shows an example of a {\gtwisted} drawing of~$K_5$.
	
	\begin{figure}
		\centering
		\includegraphics[scale=0.6,page=1]{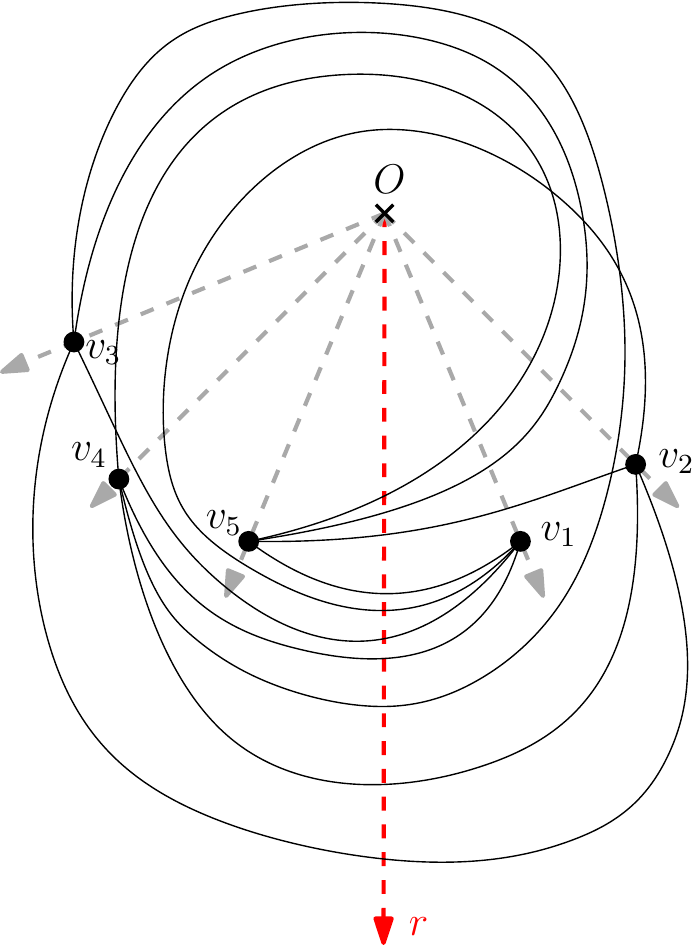}
		\caption{A {\gtwisted} drawing of $K_5$. All edges cross the (red) ray~$r$.}
		\label{fig:example_gtwisted}
	\end{figure}
	
	Generalized twisted drawings turn out to have quite surprising structural properties.
	We show some crossing properties of 
	generalized twisted drawings in Section~\ref{sec:gtwisted_crossings} and with that also prove that they always contain plane Hamiltonian paths (Theorem~\ref{thm:twisted}).
	This result is an essential ingredient for showing that any simple drawing of~$K_n$ contains $\Omega(\sqrt{n})$ pairwise disjoint edges (Theorem~\ref{thm:gen} in Section~\ref{sec:gen}), as well as a plane path of length $\Omega(\frac{\log n }{\log \log n})$ (Theorem~\ref{the:bound} in Section~\ref{sec:planepaths}).
	In Section~\ref{sec:char}, we present different characterizations of {\gtwisted} drawings that are of independent interest. 
	We conclude with an outlook on further work and open problems in Section~\ref{sec:conclusion}. 
	
	\section{Twisted Preliminaries}\label{sec:gtwisted_crossings}
	
	In this section, we show some properties of generalized twisted drawings, which will be used in the following sections. 
	
	\begin{restatable}{lem}{twist}\label{lem:twist}
		Let~$D$ be a {\gtwisted} drawing of~$K_4$, with vertices $\{v_1, v_2, v_3, v_4\}$ labeled counterclockwise around~$O$.
		Then the edges~$v_1v_3$ and~$v_2v_4$ do not cross.
	\end{restatable}

	The full proof of Lemma~\ref{lem:twist} can be found in Appendix~\ref{appendix:lem_twist}.
	\begin{proof}[Proof Sketch]
		Assume, for a contradiction, that the edge~$v_1v_3$ crosses the edge~$v_2v_4$.
		There are (up to strong isomorphism) two possibilities to draw the crossing edges~$v_1v_3$ and~$v_2v_4$, depending on whether~$v_1v_3$ crosses the (straight-line) segment from $O$ to~$v_4$ or not; cf.\ Figure~\ref{fig:prop1}.
		In both cases, there is only one way to draw $v_1v_2$ such that the drawing stays \gtwisted, yielding two regions bounded by all drawn edges ($v_1v_3$, $v_2v_4$, $v_1v_2$).
		The vertices $v_3$ and $v_4$ lie in the same region. It is well-known that every simple drawing of~$K_4$ has at most one crossing. Thus, the edge~$v_3v_4$ cannot leave this region. However, it is impossible to draw~$v_3v_4$ without leaving the region such that it is {\cmonotone} and crosses the ray~$r$
		(see the dotted arrows in Figure~\ref{fig:prop1} for necessary emanating directions of $v_3v_4$).
		\begin{figure}
			\centering
			\begin{subfigure}[b]{0.4\textwidth}
				\centering
				\includegraphics[scale=0.6,page=1]{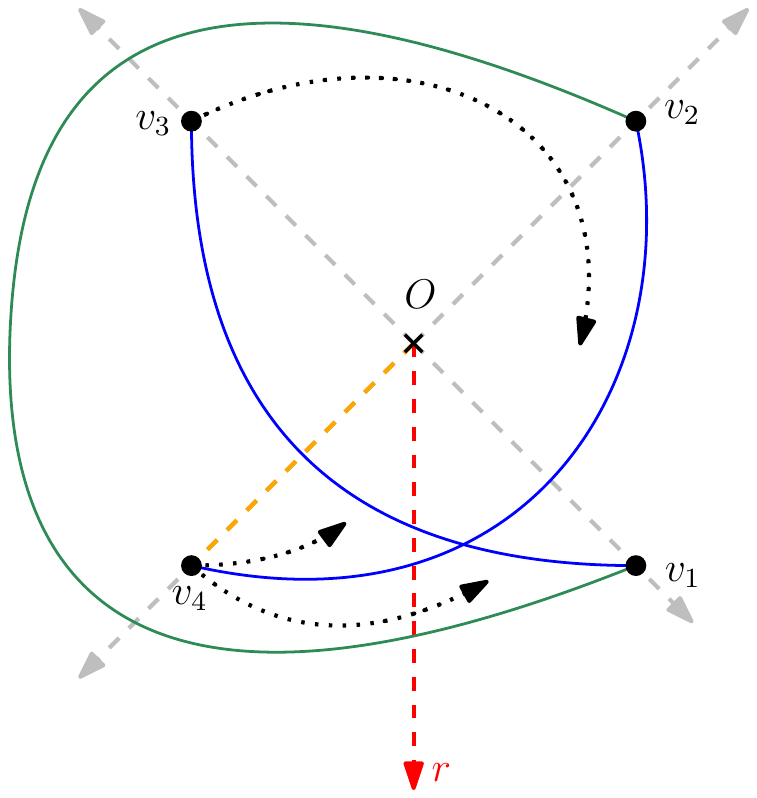}
			\end{subfigure}
			\hfill
			\begin{subfigure}[b]{0.4\textwidth}
				\centering
				\includegraphics[scale=0.6,page=2]{propertie1}
			\end{subfigure}
			\caption{The two possibilities to draw $v_1v_3$ and~$v_2v_4$ crossing and \gtwisted.
			}
			\label{fig:prop1}
		\end{figure}	
	\end{proof}
	
	Using the crossing property of Lemma~\ref{lem:twist}, it follows directly that generalized twisted drawings always contain plane Hamiltonian paths.
	
	\begin{theorem}\label{thm:twisted}
		Every {\gtwisted} drawing of~$K_n$ contains a plane Hamiltonian path. 
	\end{theorem}
	
	\begin{proof}[Proof of Theorem~\ref{thm:twisted}]
		Let $D$ be a {\gtwisted} drawing of $K_n$.
		Consider the Hamiltonian path $v_1, v_{\lceil \frac{n}{2}\rceil +1}, v_2, v_{\lceil \frac{n}{2}\rceil +2}, v_3, \ldots , v_{\lceil \frac{n}{2}\rceil -1}, v_n, v_{\lceil \frac{n}{2}\rceil}$ if $n$ is odd or the Hamiltonian path $v_1, v_{\lceil \frac{n}{2}\rceil +1}, v_2, v_{\lceil \frac{n}{2}\rceil +2}, v_3, \ldots , v_{n-1}, v_{\lceil \frac{n}{2}\rceil}, v_n$ if $n$ is even. See for example the Hamiltonian path $v_1, v_4, v_2, v_5, v_3$ in Figure~\ref{fig:example_gtwisted}. Take any pair of edges $(v_i,v_j)$ and $(v_{k}, v_{l})$ of the path, where we can assume without loss of generality that $i < j$ and $k < l$. If the two edges share an endpoint, they are adjacent and do not cross. Otherwise, if they do not share an endpoint, either $i < k < j < l$ or $k < i < l < j$ by definition of the path. In any of the two cases, $(v_i,v_j)$ and $(v_{k}, v_{l})$ cannot cross by Lemma~\ref{lem:twist}. Therefore, no pair of edges cross, and the Hamiltonian path is plane.
	\end{proof}
	
	Analogous to the proof of Theorem~\ref{thm:twisted}, one can argue that in every generalized twisted drawing of $K_n$ with $n$ odd, the Hamiltonian cycle $v_1, v_{\lceil \frac{n}{2}\rceil +1}, v_2, v_{\lceil \frac{n}{2}\rceil +2}, \ldots , v_{\lceil \frac{n}{2}\rceil -1}, v_n, v_{\lceil \frac{n}{2}\rceil}, v_1$ is plane. We strongly conjecture that every generalized twisted drawing of~$K_n$ contains a plane Hamiltonian cycle, but its structure for even $n$ is still an open problem.
	
	Theorem~\ref{thm:twisted} will be used heavily in the next two sections. 
	Further, the following statement, which has been implicitly shown in \cite{bound_2014} and \cite{triangles_together}, will be used in all remaining sections.
	For completeness, we include a proof in Appendix~\ref{appendix:quasi_c_monotone}.
	
	\begin{restatable}{lem}{quasi}\label{lem:quasi_c_monotone}
		Let $D$ be a simple drawing of a complete graph containing a subdrawing $D'$, which is a plane drawing of $K_{2,n}$. Let $A=\{a_1,a_2, \ldots, a_n\}$ and $B=\{b_1,b_2\}$ be the sides of the bipartition of $D'$. Let $D_A$ be the subdrawing of $D$ induced by the vertices of $A$. Then $D_A$ is weakly isomorphic to a c-monotone drawing. Moreover, if all edges in $D_A$ cross the edge $b_1b_2$, then $D_A$ is weakly isomorphic to a generalized twisted drawing.
	\end{restatable}
	
\pagebreak	
	\section{Disjoint Edges in Simple Drawings}\label{sec:gen}
	
	In this section, we show that every simple drawing of~$K_n$ contains at least~$\floor[\Big]{\sqrt{\frac{n}{48}}}$ pairwise disjoint edges, improving the previously known best bound of $\Omega(n^{\frac{1}{2}-\varepsilon})$, for any $\varepsilon>0$, by Ruiz-Vargas~\cite{RUIZVARGAS2017}.
	In addition to the properties of generalized twisted drawings from Section~\ref{sec:gtwisted_crossings}, we use the following theorems and observations to prove this new lower bound.
	
	\begin{theorem}[\cite{biconnected}]\label{thm:2_connected}
		For $n\ge 3$, every maximal plane subdrawing of any simple drawing of~$K_n$ is biconnected.
	\end{theorem}

	The following theorem is a direct consequence of Corollary~5 in~\cite{triangles_alone}.
	\begin{theorem}\label{thm:new_edges}
		Let~$D$ be a simple drawing of $K_n$ with $n\geq 3$. Let~$H$ be a connected plane subdrawing of~$D$ containing at least two vertices, and let~$v$ be a vertex in~$D \setminus H$. Then $D$ contains two edges incident to~$v$ that connect~$v$ with~$H$ and do not cross any edges of~$H$.
	\end{theorem}
	
	\begin{observation}\label{thm:Euler}
		For any $n \geq 3$, the number of edges in a planar graph with $n$~vertices is at most~$3n-6$.
	\end{observation}
	
	A drawing is \emph{outerplane} if it is plane, and all vertices lie on the unbounded face of the drawing. A graph is \emph{outerplanar} if it can be drawn outerplane. Outerplanar graphs have a smaller upper bound on their number of edges than planar graphs.
	
	\begin{observation}\label{thm:outerplanar}
		For any $n \geq 3$, the number of edges in an outerplanar graph with $n$~vertices is at most $2n-3$.
	\end{observation}
	
	\begin{theorem}\label{thm:gen}
		Every simple drawing of~$K_n$ contains at least~$\floor[\Big]{\sqrt{\frac{n}{48}}}$ pairwise disjoint edges.
	\end{theorem}
	
	\begin{proof}
		Let $D$ be a simple drawing of~$K_n$, and let $M$ be a maximal plane matching of $D$. If $m := |M| \geq \sqrt{\frac{n}{48}}$, then Theorem~\ref{thm:gen} holds. So assume that~$|M| < \sqrt{\frac{n}{48}}$. We will show how to find another plane matching, whose size is at least $\floor{\sqrt{\frac{n}{48}}}$. 
		
		The overall idea is the following:
		Let $H$ be a maximal plane subdrawing of $D$ whose vertex set is exactly the vertices matched in $M$ and that contains~$M$.
		We will find a face~$f$ in~$H$ that contains much more unmatched vertices inside than matched vertices on its boundary. Then we will show that there exists a subset of the vertices inside that face, which induces a subdrawing of $D$ that is weakly isomorphic to a {\gtwisted} drawing and contains enough vertices to guarantee the desired size of the plane matching.
		
		We start towards finding the face $f$.
		By Theorem~\ref{thm:2_connected}, $H$ is biconnected. Thus, $H$ partitions the plane into faces, where the boundary of each face is a simple cycle. Note that the vertices of~$H$ are exactly the vertices that are matched in $M$, and the vertices inside faces are the vertices that are unmatched in~$M$.
		Let $U$ be the set of vertices of $D$ that are not matched by any edge of $M$.
		We denote the set of vertices of $U$ inside a face $f_i$ by $U(f_i)$, the number of vertices in $U(f_i)$ by $u(f_i)$,
		and the number of vertices on the boundary of the face $f_i$ by $|f_i|$.
		
		We next show that there exists a face $f$ of $H$ such that $u(f) \geq {\frac{\sqrt{48n}}{12}}|f|$.
		Assume for a contradiction that for every face $f_i$ it holds that
		\begin{equation*}
			u(f_i) <  {\frac{\sqrt{48n}}{12}}|f_i|.
		\end{equation*}
		There are exactly $n-2m$ unmatched vertices. As every unmatched vertex is in the interior of a face of $H$ (that might be the unbounded face), we can count the unmatched vertices by summing over the number of vertices in each face (including the unbounded face). Thus,
		\begin{equation}n-2m \leq \sum_{f_i}u(f_i) < {\frac{\sqrt{48n}}{12}}\sum_{f_i}|f_i|.\end{equation}
		The number of edges in $H$ is $\frac{1}{2}\sum_{f_i}|f_i|$. Since~$H$ is plane, we can use Observation~\ref{thm:Euler} to bound the number of edges of~$H$ by $3n'-6$, where $n'$ is the number of vertices in~$H$. As the vertices of~$H$ are exactly the matched vertices, their number is $n'=2m$. Hence, \begin{equation*}\sum_{f_i}|f_i|\leq 6 \cdot 2m-12.\end{equation*}
		From $m < \sqrt{\frac{n}{48}}$ it follows that \begin{equation}\sum_{f_i}|f_i| < 12\sqrt{\frac{n}{48}}-12\end{equation} and \begin{equation}n-2 \sqrt{\frac{n}{48}} < n-2m.\end{equation}
		Putting equations (1) to (3) together we obtain that \begin{equation*}n-2 \sqrt{\frac{n}{48}} <   {\frac{\sqrt{48n}}{12}} (12\sqrt{\frac{n}{48}}-12)=n - \sqrt{48n}.\end{equation*}
		
		However, this inequality cannot be fulfilled by any $n \geq 0$. Thus, there exists at least one face $f_i$ with $u(f_i) \geq  {\frac{\sqrt{48n}}{12}}|f_i|$. We call that face $f$. (If there are several such faces, we take an arbitrary one of them and call it~$f$.)
		
		As a next step, we will find two vertices on the boundary of~$f$ to which many vertices inside~$f$ are connected via edges that do not cross each other or~$H$.
		From $f$ and the set $U(f)$, we construct a plane subdrawing $H'$ as follows; cf.~Figure~\ref{fig:edges}~(left). We add the vertices and edges on the boundary of $f$. 
		Then we iteratively add all the vertices in $U(f)$, where for each added vertex $v$ we also add two edges of $D$ incident to $v$ such that the resulting drawing stays plane. 
		Two such edges exist by Theorem~\ref{thm:new_edges}. Since the matching~$M$ is maximal, any edges between two unmatched vertices must cross at least one edge of $M$ and thus must cross the boundary of $f$. Hence, no edge in $H'$ can connect two vertices of $U(f)$ (as they are unmatched). Consequently, every vertex in $U(f)$ is connected in $H'$ to exactly two vertices that both lie on the boundary of $f$.
		
		\begin{figure}[htb]
			\centering
			\includegraphics[page=15, scale=0.92]{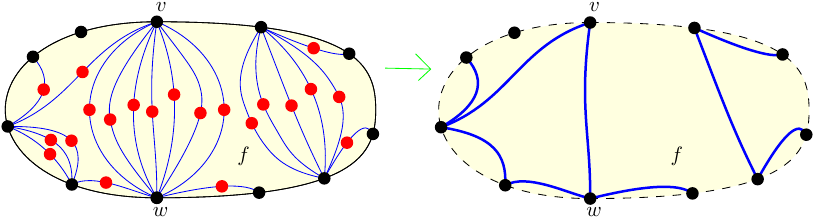}
			\caption{Left: The face $f$ in $H$ containing the plane drawing $H'$ (blue lines) inside. Right: We can obtain an outerplane drawing from $H'$ by interpreting bundles of edge pairs incident to the same black vertices as plane edges.}
			\label{fig:edges}
		\end{figure}

		We consider the edges in $H'$ that connect a vertex in $U(f)$ as a pair of edges. Every edge in such a pair is contained in exactly one pair, since it is incident to exactly one unmatched vertex. Thus, we can see every such pair of edges as one \emph{long edge} incident to two vertices on the boundary of~$f$. If several of those long edges have the same endpoints, we call them a bundle of edges; see Figure~\ref{fig:edges}~(right).
		
		From the long edges, we can define a graph $G'$ as follows. The vertices of $G'$ are the vertices of $D$ that lie on the boundary of~$f$. Two vertices $u$ and $v$ are connected in $G'$ if there is at least one long edge in $H'$ that connects them. By the definition of long edges, $G'$ is outerplanar (as can be observed in Figure~\ref{fig:edges}~(right)).  	
		Note that every unmatched vertex in $U(f)$ defines a long edge, so the number of long edges is $u(f) \geq \frac{\sqrt{48n}}{12}|f|$.
		From Observation~\ref{thm:outerplanar}, it follows that $G'$ has at most $2|f|-3$ edges. As a consequence,
		there is a pair of vertices on the boundary of $f$ such that the number of long edges in its bundle is at least
		\begin{equation*}
			{\frac{1}{(2 |f|-3)}}\frac{\sqrt{48n}}{12}|f|> 	\frac{\sqrt{48n}}{{24}}.
		\end{equation*}
		This implies that there are two vertices, say $v$ and $w$, to which more than $\frac{\sqrt{48n}}{{24}}$ vertices inside~$f$ have two plane incident edges. We call the set of vertices in $U(f)$ that have plane edges to both vertices~$v$ and~$w$ the set~$U_{vw}$. This set is marked in Figure~\ref{fig:gtwisted_subset}~(left). We denote the subdrawing of $D$ induced by~$U_{vw}$ by $D_{vw}$; see Figure~\ref{fig:gtwisted_subset}~(right).
		
		\begin{figure}[htb]
			\centering
			\includegraphics[page=10, scale=0.92]{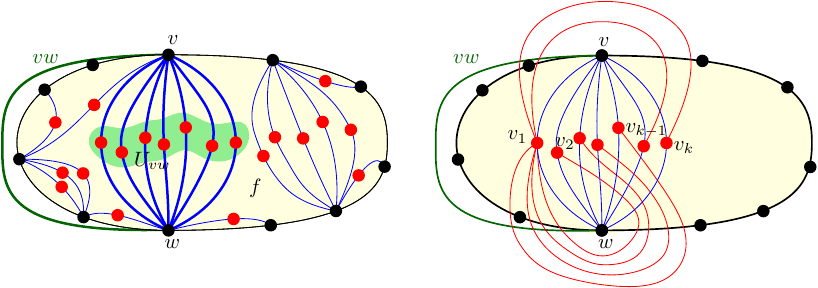}
			\caption{The subdrawing $D'$ induced by $U_{vw}$ and the edges in $D_{vw}$. Left: The set~$U_{vw}$. Right: The edges adjacent to the leftmost vertex, $v_1$, are drawn (in red).}
			\label{fig:gtwisted_subset}
		\end{figure}
		
		We show that all edges between vertices in $U_{vw}$ cross the edge $vw$. 
		Let $x$ and $y$ be two vertices of $D_{vw}$. Let $R_1$ be the region bounded by the edges $xv$, $vy$, $yw$, and $wx$ that lies inside the face~$f$; see Figure~\ref{fig:is_twisted}. We show that $xy$ and  $vw$ lie completely outside~$R_1$.
		The edge $xy$ has to lie either completely inside or completely outside~$R_1$, because it is adjacent to all edges on the boundary of~$R_1$. As $M$ is maximal and the edge~$xy$ connects two unmatched vertices, it has to cross at least one matching edge. Thus, $xy$ has to lie completely outside~$R_1$. (There can be no matching edges in~$R_1$, as $R_1$ is contained inside the face~$f$.)
		As $H$ is a maximal plane subdrawing, $vw$ cannot lie inside the face $f$ and thus has to be outside $R_1$.
		Since both edges $vw$ and  $xy$ lie completely outside $R_1$ and the vertices along the boundary of $R_1$ are sorted $vxwy$, the two edges have to cross. Thus, all edges of~$D_{vw}$ cross the edge~$vw$.

		\begin{figure}[hbt]
			\centering
			\includegraphics[page=16, scale=0.92]{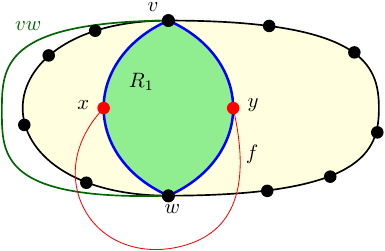}
			\caption{The edge $xy$ has to cross the edge $vw$.}
			\label{fig:is_twisted}
		\end{figure}
		
		Since the edges from vertices in $U_{vw}$ to $v$ and $w$ are plane, it follows from Lemma~\ref{lem:quasi_c_monotone} that $D_{vw}$ is weakly isomorphic to a generalized twisted drawing. Thus, $D_{vw}$ contains at least $\lfloor \frac{1}{2}\frac{\sqrt{48n}}{{24}}\rfloor$ pairwise disjoint edges by Theorem~\ref{thm:twisted}. Hence, $D$ contains at least $\lfloor\sqrt{\frac{n}{48}}\rfloor$
		pairwise disjoint edges.
	\end{proof}
	
\pagebreak	
	\section{Plane Paths in Simple Drawings}\label{sec:planepaths}
	
	In the previous section, we used generalized twisted drawings to improve the lower bound on the number of disjoint edges in simple drawings of~$K_n$. In this section, we show that generalized twisted drawings are also helpful to improve the lower bound on the length of the longest path in such drawings, where the length of a path is the number of its edges, to $\Omega(\frac{\log n }{\log \log n})$. This improves the previously known best bound of $\Omega((\log n )^{\frac{1}{6}})$, which follows from a result of Pach, Solymosi, and T{\'o}th~\cite{bound_2003}.
	
	\begin{theorem}\label{the:bound}
		Every simple drawing $D$ of $K_n$ contains a plane path of length~$\Omega(\frac{\log n }{\log \log n})$.
	\end{theorem}

	To prove the new lower bound, we first show that all c-monotone drawings on $n$ vertices contain either a generalized twisted drawing on $\sqrt{n}$ vertices or a drawing weakly isomorphic to an x-monotone drawing on $\sqrt{n}$ vertices. We know that drawings weakly isomorphic to generalized twisted drawings or x-monotone drawings  contain plane Hamiltonian paths (by Theorem~\ref{thm:twisted} and Observation~\ref{cor:monotone} below). We conclude that c-monotone drawings contain plane paths of the desired size. We then show that every simple drawing of the complete graph contains either a c-monotone drawing or a plane $d$-ary tree. With easy observations about the length of the longest path in $d$-ary trees and by putting all results together, we obtain that every simple drawing $D$ of $K_n$ contains a plane path of length~$\Omega(\frac{\log n }{\log \log n})$.
	
	\subsection{Plane Paths in C-Monotone Drawings}
	
	A simple drawing is $x$-monotone if any vertical line intersects any edge of the drawing at most once (see Figure~\ref{fig:xmonotone}). This family of drawings has been studied extensively in the literature 
	(see for example~\cite{shell1, monotone_algo,  monotone_characterization, monotone_tutte, monotone_crossings}).
	By definition, c-monotone drawings in which there exists a ray emanating from $O$, which crosses all edges of the drawing, are generalized twisted.
	In contrast, consider a c-monotone drawing $D$ such that there exists a ray $r$ emanating from~$O$ that crosses no edge of $D$. Then it is easy to see that $D$ is strongly isomorphic to an $x$-monotone drawing.
	(A c-monotone drawing on the sphere can be cut along the ray $r$ and the result drawn on the plane such that all rays are vertical lines and the ray $r$ is to the very left of the drawing.)
	Figure~\ref{fig:cmonotone} shows a c-monotone drawing $D$ of $K_5$ where no edge crosses the ray $r$, and Figure~\ref{fig:xmonotone} shows an $x$-monotone drawing of $K_5$ strongly isomorphic to $D$. We will call simple drawings that are strongly isomorphic to $x$-monotone drawings \emph{monotone} drawings. In particular, any c-monotone drawing for which there exists a ray emanating from~$O$ that crosses no edge of the drawing is monotone.
	
	\begin{figure}
		\centering
		\begin{subfigure}[b]{0.4\textwidth}
			\centering
			\includegraphics[scale=0.6,page=5]{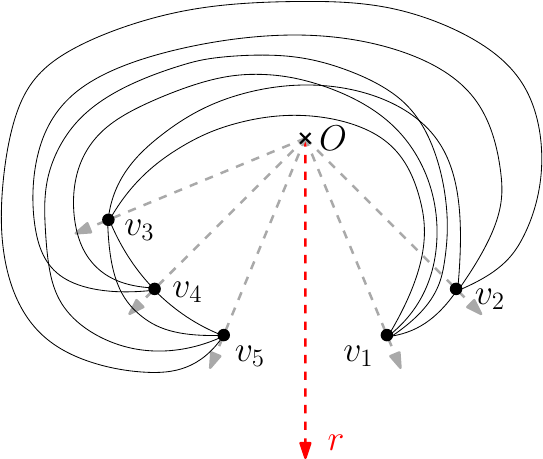}
			\caption{A c-monotone drawing $D$ of $K_5$ such that the ray $r$ crosses no edge of $D$.}
			\label{fig:cmonotone}
		\end{subfigure}
		\hfill
		\begin{subfigure}[b]{0.4\textwidth}
			\centering
			\includegraphics[scale=0.6,page=4]{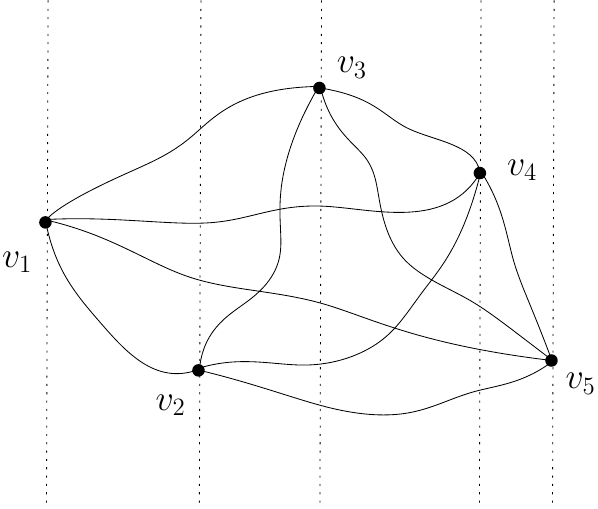}\caption{An $x$-monotone drawing of $K_5$ strongly isomorphic to $D$ of Figure~\ref{fig:cmonotone}.}	\label{fig:xmonotone}
		\end{subfigure}
		\caption{Two strongly isomorphic monotone drawings of $K_5$.}
		\label{fig:monotone}
	\end{figure}	
	
	It is well-known that any $x$-monotone drawing of $K_n$ contains a plane Hamiltonian path. For instance, assuming that the vertices are ordered by increasing $x$-coordinates, the set of edges $v_1v_2, v_2v_3 \ldots , v_{n-1}v_n$ form a plane Hamiltonian path.
	
	\begin{observation}\label{cor:monotone}
		Every monotone drawing of $K_n$ contains a plane Hamiltonian path.
	\end{observation}

	We will show that c-monotone drawings contain plane paths of size $\sqrt{n}$, by showing that any c-monotone drawing of $K_n$ contains a subdrawing of $K_{\sqrt{n}}$ that is either generalized twisted or monotone.
	To do so, we will use Dilworth's Theorem on chains and anti-chains in partially ordered sets. 
	A \emph{chain} is a subset of a partially ordered set such that any two distinct elements are comparable. An \emph{anti-chain} is a subset of a partially ordered set such that any two distinct elements are incomparable.
	
	\begin{theorem}[Dilworth's Theorem, \cite{dilworth}]\label{thm:dilworth}
		Let $P$ be a partially ordered set of at least \mbox{$(s\!-\!1)(t\!-\!1)\!+\!1$}~elements. Then $P$ contains a chain of size $s$ or an antichain of size~$t$.
	\end{theorem}
	
	\begin{restatable}{theorem}{thmsubdrawing}\label{thm:subdrawing}
		Let $s,t$ be two integers, $1\leq s,t\leq n$, such that $(s-1)(t-1)+1\leq n$. Let~$D$ be a c-monotone drawing of~$K_n$. Then $D$ contains either a {\gtwisted} drawing of $K_s$ or a monotone drawing of $K_t$ as subdrawing. In particular, if $s=t= \lceil \sqrt{n} \rceil $, $D$ contains a complete subgraph $K_s$ whose induced drawing is either  {\gtwisted} or monotone.
	\end{restatable}
	
	The full proof of Theorem~\ref{thm:subdrawing} can be found in Appendix~\ref{appendix:thm_subdrawing}
	\begin{proof}[Proof Sketch]
		Without loss of generality we may assume that the vertices of $D$ appear counterclockwise around $O$ in the order $v_1,v_2,\ldots ,v_n$. Let $r$ be a ray emanating from $O$, keeping $v_1$ and~$v_n$ on different sides. We define an order, $\preceq $, in this set of vertices as follows:  $v_i\preceq v_j$ if and only if either $i=j$ or $i<j$ and the edge $(v_i,v_j)$ crosses $r$.
		
		We show that $\preceq$ is a partial order. The relation is clearly reflexive and antisymmetric. Besides, if $v_i\preceq v_j$ and $v_j\preceq v_k$, then $v_i\preceq v_k$, because $i<j$ and $j<k$ imply $i<k$, and if $v_iv_j$ and $v_jv_k$ cross $r$, then $v_iv_k$ also crosses $r$ (see Figure~\ref{fig:PartialOrder}). Hence, the relation is transitive.
		
		\begin{figure}
			\begin{subfigure}[b]{0.45\textwidth}
				\centering
				\includegraphics[scale=0.6, page=6]{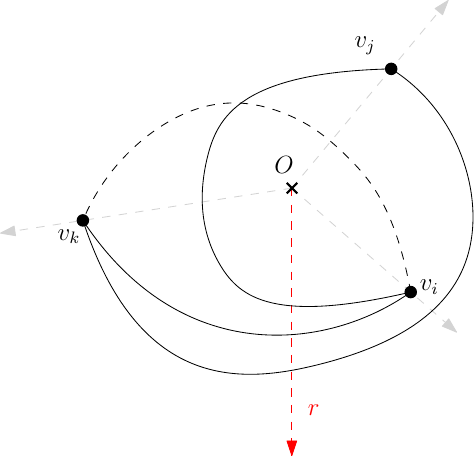}
			\end{subfigure}
			\hfill
			\begin{subfigure}[b]{0.45\textwidth}
				\centering
				\includegraphics[scale=0.6, page=8]{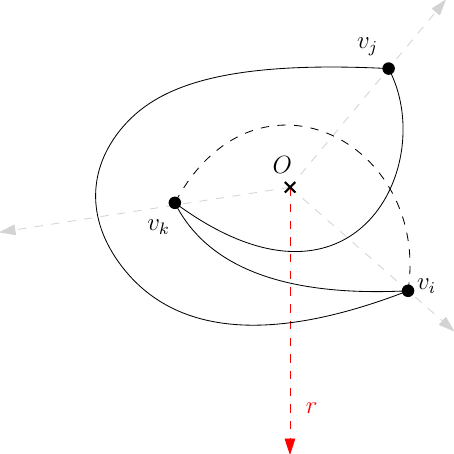}
			\end{subfigure}
			\caption{If edges $v_iv_j$ and $v_jv_k$ cross $r$ in a c-monotone drawing, then $v_iv_k$ must also cross~$r$.}\label{fig:PartialOrder}
		\end{figure}
		
		In this partial order $\preceq $, a chain consists of a subset $v_{i_1},\ldots ,v_{i_{s-1}}$ of pairwise comparable vertices, that is, a subset of vertices such that their induced subdrawing is {\gtwisted} (all edges cross $r$). An antichain, $v_{j_1},\ldots ,v_{j_{t-1}}$, consists of a subset of pairwise incomparable vertices, that is, a subset of vertices such that their induced subdrawing is monotone (no edge crosses $r$). 
		Therefore, the first part of the theorem follows from applying Theorem~\ref{thm:dilworth} to the set of vertices of $D$ and the partial order $\preceq $.
		
		Finally, observe that if $s=t\leq \lceil \sqrt{n}\rceil $, then $(s-1)(t-1)+1\leq n$. Thus, $D$ contains a complete subgraph $K_{\lceil \sqrt{n}\rceil}$ whose induced subdrawing is either  {\gtwisted} or monotone.
	\end{proof}
	
	Combining Theorems~\ref{thm:twisted} and~\ref{thm:subdrawing} with Observation~\ref{cor:monotone}, we obtain the following theorem.
	
	\begin{theorem}\label{the:pathcmonotone}
		Every c-monotone drawing of $K_n$ contains a plane path of length~$\Omega (\sqrt{n})$.
	\end{theorem}
	
	\subsection{Plane Paths in Simple Drawings}\label{sec:planepathsgen}
	
	To show that any simple drawing of $K_n$ contains a plane path of length~$\Omega(\frac{\log n }{\log \log n})$, 
	we will 
	use $d$-ary trees.
	A $d$-ary tree is a rooted tree in which no vertex has more than $d$ children. It is well-known that the height of a $d$-ary tree on $n$ vertices is $\Omega(\frac{\log n}{\log d})$. 
	
	\begin{proof}[Proof of Theorem~\ref{the:bound}]
		Let $v$ be a vertex of $D$ and let $S(v)$ be the star centered at $v$, that is, the set of edges of $D$ incident to $v$. $S(v)$ can be extended to a maximal plane subdrawing $H$ that must be biconnected by Theorem~\ref{thm:2_connected}. See Figure~\ref{fig:S_v} for a depiction of $S(v)$ and $H$.
		
		\begin{figure}[!htb]
			\centering
			\includegraphics[scale=0.5, page=7]{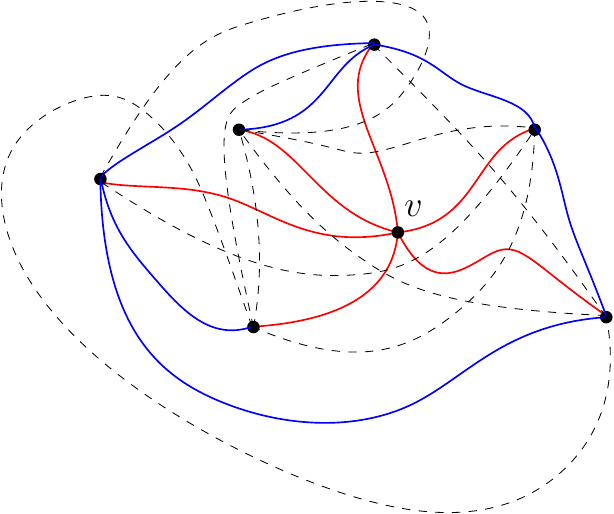}
			\caption{A simple drawing of $K_7$. The red edges show the star $S(v)$, the red and blue edges together form a maximal plane subdrawing $H$. Dashed edges are edges of $K_7$ that are not in~$H$.}\label{fig:S_v}
		\end{figure}
		
		Assume first that there is a vertex $w$ in $H\setminus v$ that has degree at least $(\log n)^2$ in $H$. Let $U_{vw}$ be the set of vertices neighboured in $H$ to both, $v$ and $w$. Note that $|U_{vw}| \ge (\log n)^2$.  The subdrawing $H'$ of $H$ consisting of the vertices in $U_{vw}$, the vertices $v$, and $w$, and the edges from $v$ to vertices in $U_{vw}$, and from $w$ to vertices in $U_{vw}$ is a plane drawing of $K_{2,|U_{vw}|}$. From Lemma~\ref{lem:quasi_c_monotone}, it follows that the subdrawing of $D$ induced by $U_{vw}$ is weakly isomorphic to a {\cmonotone} drawing.
		Therefore, by Theorem~\ref{the:pathcmonotone}, the subdrawing induced by $U_{vw}$ contains a plane path of length~$\Omega (\sqrt{|U_{vw}|}) = \Omega (\log n)$.
		
		Assume now that the maximum degree in $H\setminus v$ is less than $(\log n)^2$. Since $H$ is biconnected, $H\setminus v$ contains a plane tree $T$ of order $n-1$ whose maximum degree is at most $(\log n)^2$. Thus, considering that $T$ is rooted, the diameter of $T$ is at least $\Omega(\frac{\log n }{\log \log n})$.
		Therefore, since $T$ is plane, it contains a plane path of length at least $\Omega(\frac{\log n }{\log \log n})$ and the theorem follows.
	\end{proof}
	
	\section{Characterizing Generalized Twisted Drawings}\label{sec:char}
	In previous sections, we have seen how generalized twisted drawings were used to make progress on open problems of simple drawings.
	In addition to this, generalized twisted drawings are also interesting in their own right and have some quite surprising structural properties.
	Despite the fact that research on generalized twisted drawings is rather recent and still ongoing, there are already several interesting characteristics and structural results. 
	Some of them will be presented in this section.

	One characterization involves curves crossing every edge once. From the definition of {\gtwisted} drawing (see Figure~\ref{fig:example_gtwisted}), there always exists a simple curve that crosses all edges of the drawing exactly once (for instance, a curve that starts at $O$ and follows $r$ until it reaches a point $Z$ on $r$ in the unbounded cell). 
	In Theorem~\ref{thm:characterizations}, we show that the converse is also true. That is, every simple drawing $D$ of $K_n$ in which we can add a simple curve that crosses every edge of $D$ exactly once is weakly isomorphic to a {\gtwisted} drawing.

	Another characterization is based on what we call \emph{antipodal vi-cells}.
	For any three vertices in a simple drawing $D$ of $K_n$, the three edges connecting them form a simple cycle which we call a \emph{triangle}.
	Every such triangle partitions the plane (or sphere) into two disjoint regions which are the \emph{sides} of the triangle (in the plane a bounded and an unbounded one).
	Two cells of $D$ are called \emph{antipodal} if for each triangle of $D$, they lie on different sides.
	Further, we call a cell with a vertex on its boundary a vertex-incident-cell or, for short, a \emph{vi-cell}.
	
	By definition, every {\gtwisted} drawing $D$ contains two antipodal cells, namely, the cell containing the starting point of the ray $r$ and the unbounded cell.
	This follows from the fact that the ray $r$ crosses every edge exactly once.
	Hence, $r$ crosses the boundary of any triangle exactly three times, so the cells containing the ``endpoints'' of $r$ must be on different sides of the triangle.

	\begin{figure}[h!tb]
		\centering
		\includegraphics[page=1, scale=0.6]{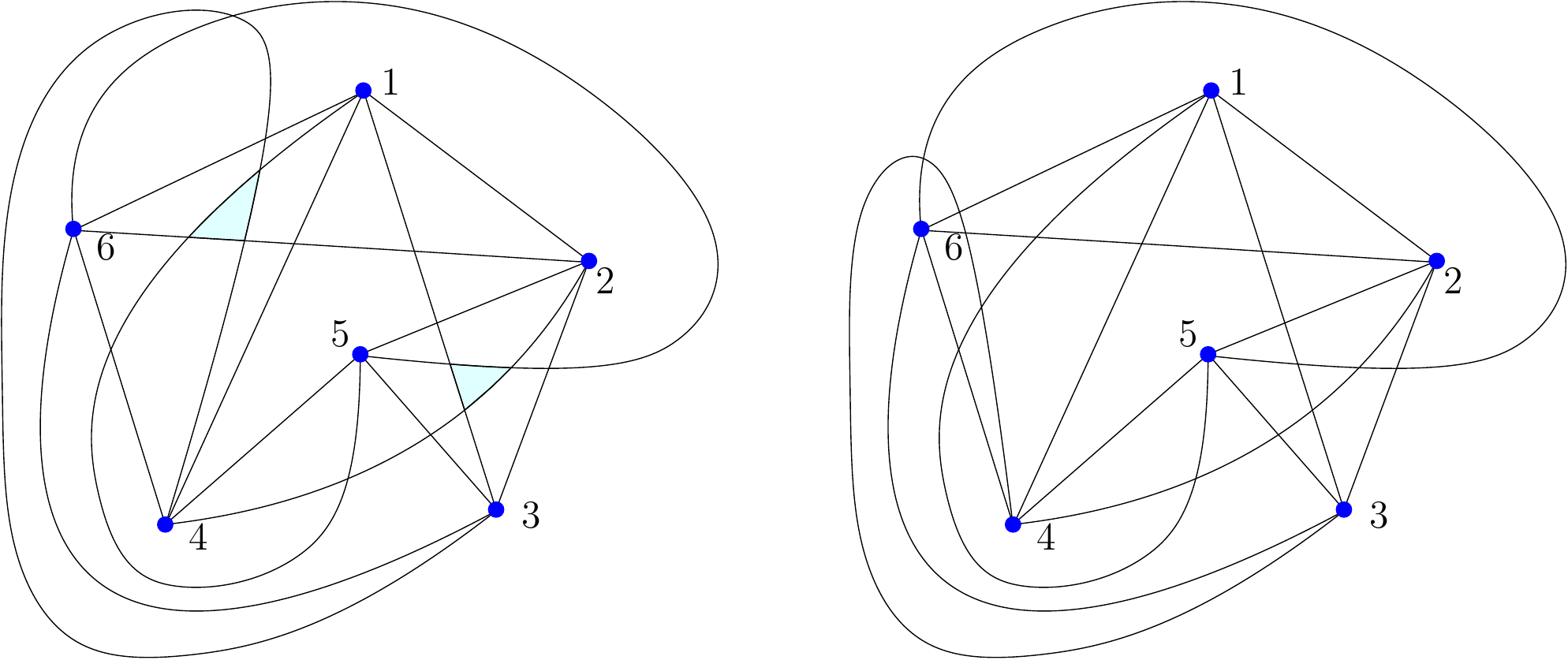}
		\caption{Two weakly isomorphic drawings of~$K_6$ that are not weakly isomorphic to any {\gtwisted} drawing. 			
		Antipodal cells are marked in blue. 
	}
		\label{fig:special_k6}
	\end{figure}
	
	It turns out that the converse (existence of two antipodal cells implies weakly isomorphic to \gtwisted) is not true. 
	Figure~\ref{fig:special_k6}~(left) shows a drawing of $K_6$ that contains two antipodal cells, but no antipodal vi-cells. 
	From Theorem~\ref{thm:characterizations} bellow it will follow that such drawings cannot be weakly isomorphic to a {\gtwisted} drawing. 
	However, we observed that for all generalized twisted drawings of $K_n$ with $n\leq 6$, both, the cell containing the startpoint of the ray $r$ and the unbounded cell,
	are vi-cells.
	Figure~\ref{fig:gtwisted_k6} shows all (up to strong isomorphism) simple drawings of~$K_6$ that are weakly isomorphic to generalized twisted drawings.
	We show that this is true in general. More than that, we show in Theorem~\ref{thm:vicells} that every drawing of~$K_n$ that is weakly isomorphic to a {\gtwisted} drawing contains a pair of antipodal vi-cells.
	In the other direction, we show
	in Theorem~\ref{thm:characterizations} that 
	every simple drawing containing a pair of antipodal vi-cells is weakly isomorphic to a {\gtwisted} drawing.
	
		\begin{figure}[htb]
		\centering
		\includegraphics[page=8, scale=0.68]{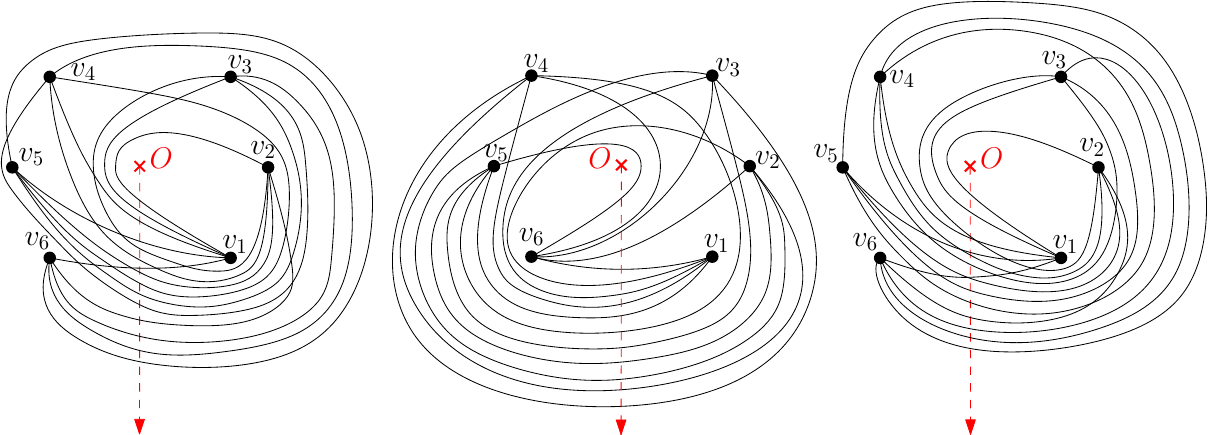}
		\caption{All different generalized twisted drawings of $K_6$ (up to weak isomorphism). The rightmost drawing is twisted.}
		\label{fig:gtwisted_k6}
	\end{figure}

	The final characterization is based on the extension of a given drawing of the complete graph to a drawing containing a spanning, plane bipartite graph that has all vertices of the original drawing on one side of the bipartition.
	From the definition of generalized twisted drawings, it follows that any genereralized twisted drawing $D$ of $K_n$ can be extended to a simple drawing $D'$ of $K_{n+2}$ including new vertices $O$ and $Z$ such that $D'$ contains a plane drawing of a spanning bipartite graph. One side of the bipartition consists of all vertices in $D$ and the other side of the bipartition consists of the new vertices $O$ and $Z$. Moreover, the edge $OZ$ crosses all edges of $D$.
	One way to add the new vertices and edges incident to them is to draw (1) the vertex $O$ at point $O$, (2) the vertex $Z$ in the unbounded cell on the ray $r$, (3) the edge $OZ$ straight-line (along the ray $r$), (4) edges from $O$ to the vertices of $D$ straight-line (along the inner segment of the rays crossing through the vertices), and (5) edges from $Z$ to the vertices of $D$ first far away in a curve and the final part straight-line (along the outer segment of the rays crossing through the vertices). The converse, that every drawing that can be extended like this is weakly isomorphic to a {\gtwisted} drawing, has already been shown in Lemma~\ref{lem:quasi_c_monotone}.
	
	We show the following characterizations.
	
	\begin{theorem}[Characterizations of generalized twisted drawings]\label{thm:characterizations}
		Let $D$ be a simple drawing of $K_n$. Then, the following properties are equivalent.
		\begin{enumerate}[leftmargin=*,label={\emph{Property~\arabic*}}]
			\item\label{char:weak_iso} $D$ is weakly isomorphic to a generalized twisted drawing.
			\item\label{char:antipodal} $D$ contains two antipodal vi-cells.
			\item\label{char:curve} $D$ can be extended by a simple curve $c$ such that $c$ crosses every edge of $D$ exactly once.
			\item\label{char:bipartite} $D$ can be extended by two vertices, $O$ and $Z$, and edges incident to the new vertices such that $D$ together with the new vertices and edges is a simple drawing of $K_{n+2}$, the edge $OZ$ crosses every edge of $D$, and no edge incident to $O$ crosses any edge incident to $Z$. 		
		\end{enumerate}
	\end{theorem}
	
	To prove Theorem~\ref{thm:characterizations}, we will first show that \ref{char:weak_iso} implies \ref{char:antipodal} (Theorem~\ref{thm:vicells}). We next show that \ref{char:antipodal} implies \ref{char:curve} (Theorem~\ref{thm:antipodal_then_gtwisted}). Then, we show that \ref{char:curve} implies \ref{char:bipartite} (Theorem~\ref{thm:curve_then_gtwisted}). By Lemma~\ref{lem:quasi_c_monotone}, \ref{char:bipartite} implies \ref{char:weak_iso}. Thus, all properties are equivalent. In a full version of this work, we will extend the theorem to show that also strong isomorphism to a generalized twisted drawing is equivalent to the properties of Theorem~\ref{thm:characterizations}. We show this by proving that any simple drawing of $K_n$ fulfilling \ref{char:bipartite} is strongly isomorphic to a generalized twisted drawing. However, the reasoning for strong isomorphism is quite lengthy and would exceed the space constraints of this submission.
	
	In the remaining parts of this section, we will show sketches of the proofs of the above mentioned theorems. The full proofs can be found in the Appendix (Theorem~\ref{thm:vicells} in Appendix~\ref{appendix:vicell}, Theorem~\ref{thm:antipodal_then_gtwisted} in Appendix~\ref{appendix:antipodal_gtwisted}, and Theorem~\ref{thm:curve_then_gtwisted} in Appendix~\ref{appendix:char_curve}).

	\begin{restatable}{theorem}{thmvicells}\label{thm:vicells}
		Every simple drawing of $K_n$ which is weakly isomorphic to a {\gtwisted} drawing of $K_n$, with $n\ge 3$, contains a pair of antipodal vi-cells. In {\gtwisted} drawings the cell containing $O$ and the unbounded cell form such a pair.
	\end{restatable}

	\begin{proof}[Proof sketch]
		\begin{figure}[h!tb]
			\centering
			\includegraphics[page=37,scale=0.6]{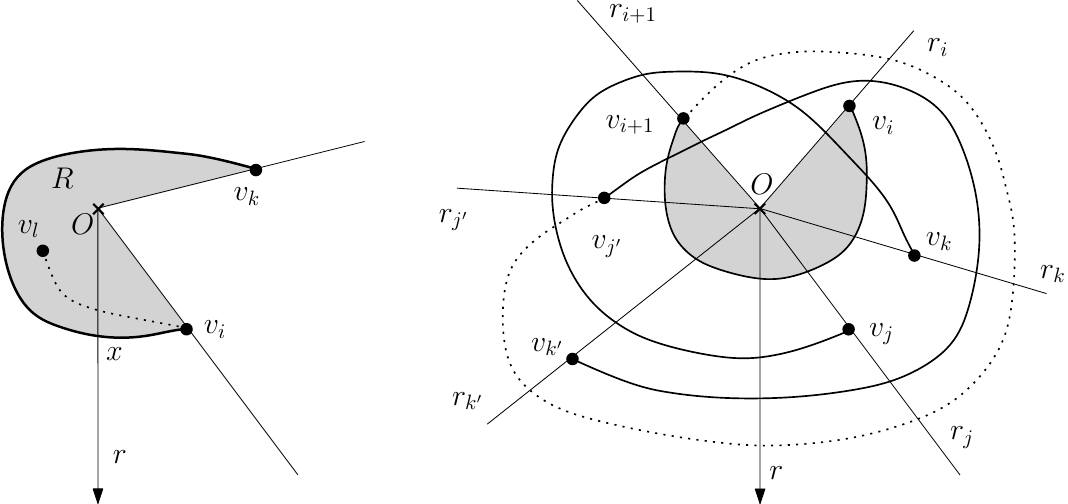}
			\caption{Left: If there is a vertex $v_l$ in $R$, it cannot be connected to $v_i$ without crossing $r$ before~$x$. Right: If the edge $v_jv_k$ crosses the segment $\overline{Ov_i}$ and the edge $v_{j'}v_{k'}$ crosses the segment $\overline{Ov_{i+1}}$, then there is no way of connecting $v_{i+1}$ and $v_{j'}$.}
			\label{fig:vicells}
		\end{figure}
		
		We first show that every {\gtwisted} drawing $D$ of $K_n$, with $n\ge 3$, contains a pair of antipodal vi-cells, where $O$ lies in a cell of that pair.
		Let $c$ be the segment $OZ$, where $Z$ is a point on $r$ in the unbounded cell. By definition of {\gtwisted}, $c$ crosses every edge of $D$ once, so $O$ and $Z$ are in two antipodal cells $C_1$ and $C_2$, respectively.
		
		To prove that $C_1$ is a vi-cell, we use the following properties. First, if we take the first edge~$v_iv_k$ that crosses~$c$ (as seen from~$O$) at point~$x$, then we can prove that $k=i+1$ and the bounded region~$R$ defined by the edge~$v_iv_{i+1}$ and the segments~$\overline{Ov_i}$ and~$\overline{Ov_{i+1}}$ is empty (see Figure~\ref{fig:vicells}, left). Second, using this empty region we can prove that $D$ cannot contain simultaneously an edge~$v_jv_k$ crossing $\overline{Ov_i}$ and another edge~$v_{j'}v_{k'}$ crossing~$\overline{Ov_{i+1}}$ (see Figure~\ref{fig:vicells}, right). Therefore, at least one of the segments~$\overline{Ov_i}$ and~$\overline{Ov_{i+1}}$ is uncrossed, and~$O$ necessarily lies in a vi-cell (with either $v_i$ or $v_{i+1}$ on the boundary). Finally, arguing on the last edge crossing $c$ and the unbounded cell, we can show that $Z$ also lies in a vi-cell.
		
		To show that also every drawing which is weakly isomorphic to a generalized twisted drawing contains a pair of antipodal vi-cells, we use Gioan's Theorem~\cite{arroyo,gioan}. By Gioan's Theorem, any two weakly isomorphic drawings of $K_n$ can be transformed into each other with a sequence of triangle-flips and at most one reflection of the drawing. A \emph{triangle-flip} is an operation which transforms a triangular cell $\triangle$ that has no vertex on its boundary by moving one of its edges across the intersection of the two other edges of~$\triangle$. We show that if a drawing $D_1$ contains two antipodal vi-cells, then after performing a triangle flip on $D_1$, the resulting drawing $D_2$ still has two antipodal vi-cells. The main argument is that triangle-flips are only applied to cells without vertices on their boundary, and thus the antipodality of the vi-cells cannot change.
	\end{proof}
	
	\begin{figure}
		\centering
		\includegraphics[page=8, scale=0.9]{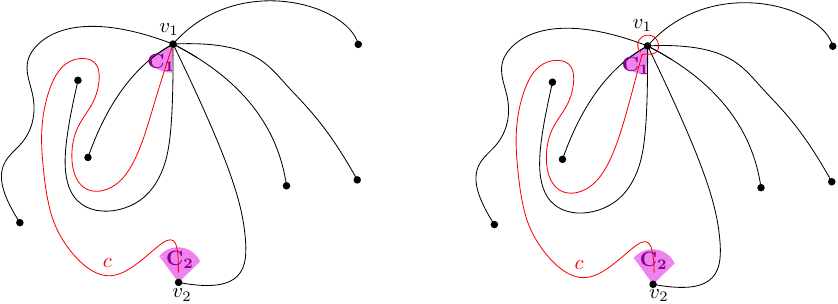}
		\caption{Building a curve such that it crosses every edge of $D$ once and its endpoints do not lie on any edges or vertices of $D$.}\label{fig:curvec}
	\end{figure}

	\begin{figure}
		\centering
		\includegraphics[page=10, scale=0.91]{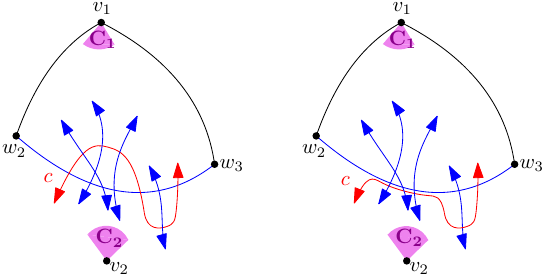}
		\caption{Decreasing the number of crossings between $c$ and the edge $w_2w_3$.}
		\label{fig:decreasing}
	\end{figure}

	\begin{restatable}{theorem}{thmantipodal}\label{thm:antipodal_then_gtwisted}
		In any simple drawing $D$ of $K_n$ that contains a pair of antipodal vi-cells, it is possible to draw a curve $c$ that crosses every edge of~$D$ exactly once.
	\end{restatable}
	
	\begin{proof}[Proof sketch]
		
		Let $(C_1,C_2)$ be a pair of antipodal vi-cells of $D$. Let $v_1$ be a vertex on the boundary of $C_1$ and $v_2$ a vertex on the boundary of $C_2$.
		We construct the curve as follows: First, we draw a simple curve $c$ from $C_1$ to $C_2$ such that (1) it emanates from $v_1$ in $C_1$ and ends in $C_2$ very close to $v_2$, (2) does not cross any edge incident to $v_1$, (3) only intersects edges of $D$ in proper crossings, and (4) has the minimum number of crossings with edges of $D$ among all curves that fulfill (1), (2) and (3). This curve $c$ always exists since $S(v_1)$ is a plane drawing that has only a face in which both $v_1$ and $v_2$ lie (see Figure~\ref{fig:curvec}, left).
		
		Then, we prove that $c$ crosses every edge $w_2w_3$ in $D$ that is not incident to $v_1$ exactly once. On the one hand, since $c$ connects two antipodal cells, the endpoints of $c$ have to be on two different sides of the triangle $T$ formed by $v_1$, $w_2$ and $w_3$. Thus, $c$ has to cross $w_2w_3$ an odd number of times because it does not cross $S(v_1)$ and must cross the boundary of $T$ an odd number of times. On the other hand, if $c$ crosses $w_2w_3$ at least three times, then we can prove that $c$ can be redrawn as shown in Figure~\ref{fig:decreasing}, decreasing the number of crossings, which contradicts (4). Therefore, $c$ crosses every edge $w_2w_3$ at most twice and, consequently, only once.
		
		Finally, we change the end of $c$ from $v_1$ to a point in $C_1$ in the following way (see  Figure~\ref{fig:curvec}, right). From some
		point of $c$ sufficiently close to $v_1$ and inside $C_1$, we reroute $c$ by going around $v_1$  
such that only the edges incident to $v_1$ are crossed, and end at a point in~$C_1$.
	\end{proof}

	\begin{figure}
		\centering
		\includegraphics[page=1, scale=0.67]{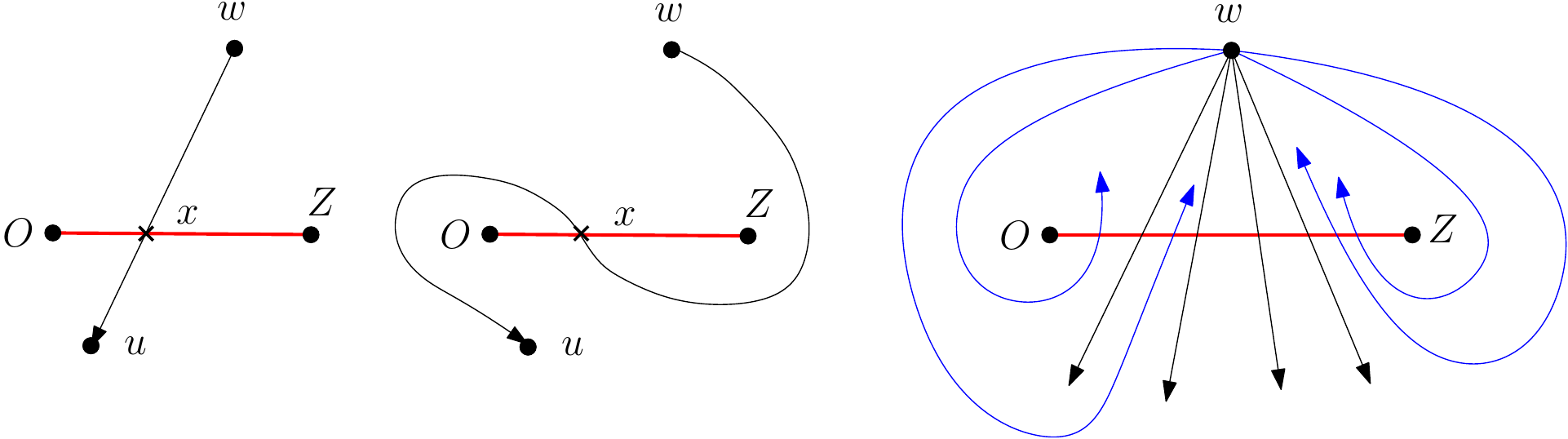}
		\caption{Top and bottom edges. For simplicity, the curve $OZ$ is drawn as a horizontal line. Left: A top edge $wu$. Centre: A bottom edge $wu$. Right: The (black) top and (blue) bottom edges of $S(w)$.}\label{fig:topbottom}
	\end{figure}

	\begin{restatable}{theorem}{thmcurvegtiwsted}\label{thm:curve_then_gtwisted}
		Let~$D$ be a simple drawing of $K_n$ in which it is possible to draw a simple curve $c$ that crosses every edge of~$D$ exactly once. Then, $D$ can be extended by two vertices $O$ and $Z$ (at the position of the endpoints of the curve), and edges incident to those vertices 
		such that the obtained drawing is a simple drawing of $K_{n+2}$, no edge incident to $O$ crosses any edge incident to $Z$, and all edges in $D$ cross the edge $OZ$.
	\end{restatable}
	
	\begin{proof}[Proof sketch]
		Let $c=OZ$ be the curve crossing every edge of $D$ once, oriented from $O$ to~$Z$. Let $wu$ be an edge of $D$, oriented from $w$ to $u$, crossing $OZ$ at a point $x$. We say that $wu$ is a \emph{top} (respectively \emph{bottom}) edge if the clockwise order of $w, Z, u$ and $O$ around $x$ is $w, Z, u, O$ (respectively $w, O, u, Z$). See Figure~\ref{fig:topbottom}. With these definitions, we can prove that there is a vertex $w_1$ in $D$ such that all the oriented edges emanating from $w_1$ are top in relation to~$c$. Thus, by removing $w_1$ and all its incident edges from $D$, there is a vertex $w_2$ in the new drawing such that all its incident edges are top, and so on. As a consequence, there is a \emph{natural} order $w_1, w_2, \ldots , w_n$ of the vertices of $D$ such that for any vertex $w_i$, the edges $w_iw_j$ with $j>i$ are top, and the edges $w_iw_j$ with $j<i$ are bottom.
		
		\begin{figure}
			\centering
			\includegraphics[page=2, scale=0.50]{curvesketch}
			\caption{Building the (dashed) edges $w_iO$ and $w_iZ$.}\label{fig:edgeswi}
		\end{figure}

		Given the natural order $w_1, w_2, \ldots , w_n$, our construction of the extended drawing is as follows. Let $D_0'$ be the simple drawing formed by the vertices and edges of $D$, $O$ and $Z$ as new vertices, and $c$ as the edge connecting $O$ and $Z$. From $D_0'$, we build new drawings $D_1', D_2', \ldots , D'_n$, by adding in step $i$ the edges $w_iO$ and $w_iZ$. These two edges are added very close to some edges in $D'_{i-1}$. Figure~\ref{fig:edgeswi} illustrates how these two edges are added in each step.
		
		In the first step, the edge $Ow_1$ follows the curve $OZ$ until the crossing point between $OZ$ and the first top edge $w_1u$ emanating from $w_1$, and then it follows this top edge until reaching $w_1$. The edge $Zw_1$ is built in an analogous way, taking the last top edge emanating from $w_1$. See Figure~\ref{fig:edgeswi} top-left. For $i=2, \ldots , n-1$, in step $i$ we do different constructions depending on whether the first and last top edges of $S(w_i)$ cross the edges $w_{i-1}O$ and $w_{i-1}Z$. If the first top edge $w_iu_1$ crosses $w_{i-1}O$ at a point $x$ and the last top edge $w_iu_k$ crosses $w_{i-1}Z$ at a point $y$ (see Figure~\ref{fig:edgeswi} top-right), then $Ow_i$ follows $Ow_{i-1}$ until $x$, and then it follows $u_1w_i$ until $w_i$. The edge $Zw_i$ is built following $Zw_{i-1}$ until $y$ and then following $u_kw_i$. On the contrary, if the first and the last top edges of $S(w_i)$ only cross one of $w_{i-1}O$ and $w_{i-1}Z$, say $w_{i-1}Z$ (see Figure~\ref{fig:edgeswi} bottom-left), then $Ow_i$ follows $OZ$ until the crossing point between $OZ$ and the last bottom edge of $S(w_i)$, and then it follows this bottom edge until $w_i$. The edge $Zw_i$ is built as in the first step, using the last top edge of $S(w_i)$. In the last step, we build $Ow_n$ and $Zw_n$ as in the first step, but using the first and the last bottom edges of $S(w_n)$ instead of the first and last top edges. See Figure~\ref{fig:edgeswi} bottom-right.
		
		By a detailed analysis of cases, we can prove for $i=1, \ldots , n$ that $D'_i$ is a simple drawing such that no edge incident to $O$ crosses any edge incident to $Z$. Therefore,  $D'_n$ is the drawing of $K_{n+2}$ satisfying the required properties.
	\end{proof}

	\section{Conclusion and Outlook}\label{sec:conclusion}
	Generalized twisted drawings have a suprisingly rich structure and many useful properties.
	We showed several of those properties in Section~\ref{sec:gtwisted_crossings} and different characterizations of generalized twisted drawings in Section~\ref{sec:char}. We have proven in Section~\ref{sec:gtwisted_crossings} that every generalized twisted drawing on an odd number of vertices contains a plane Hamiltonian cycle, and therefore one especially interesting open question is the following.
	
	\begin{conjecture}
		Every generalized twisted drawing of~$K_n$ contains a plane Hamiltonian cycle.
	\end{conjecture}
	
	Using properties of generalized twisted drawings has turned out to be helpful for investigating simple drawings in general. We first improved the lower bound on the number of disjoint edges in simple drawings of~$K_n$ to $\Omega(\sqrt{n})$ (Section~\ref{sec:gen}). Then generalized twisted drawings played the central role to improve the lower bound on the length of plane paths contained in every simple drawing of~$K_n$ to $\Omega(\frac{\log n }{\log \log n})$ (Section~\ref{sec:planepaths}).
	
	On the other hand, from Theorem~\ref{thm:antipodal_then_gtwisted} it immediately follows that no drawing that is weakly isomorphic to a generalized twisted drawing can contain three interior-disjoint triangles (since the endpoints of the curve crossing every edge once must be on opposite sides of every triangle, the maximum number of interior-disjoint triangles is two). Up to strong isomorphism, there are only two simple drawings of $K_4$. The plane drawing contains three interior-disjoint triangles. Thus, (up to strong isomorphism) the only drawing of $K_4$ that is weakly isomorphic to a generalized twisted drawing, is the drawing with a crossing. Hence, in every generalized twisted drawing all subdrawings induced by $4$ vertices contain a crossing and thus every generalized twisted drawing is crossing maximal. Up to strong isomorphism, there are two crossing maximal drawings of $K_5$: the convex drawing of $K_5$ and the twisted drawing of $K_5$. Since the convex drawing contains three interior-disjoint triangles, the only (up to strong isomorphism) drawing of $K_5$ that is weakly isomorphic to a generalized twisted drawing is the twisted drawing of $K_5$ (that is drawn generalized twisted in Figure~\ref{fig:example_gtwisted}).
	
	It is part of our ongoing work to show that for $n \geq 7$, a drawing is weakly isomorphic to a generalized twisted drawing if and only if all subdrawings induced by five vertices are weakly isomorphic to the twisted $K_5$.
	Interestingly, the $n \geq 7$ is necessary as 
	there is a drawing with $6$ vertices that contains only twisted drawings of $K_5$ but is not weakly isomorphic to a generalized twisted drawing (see the drawings in Figure~\ref{fig:special_k6}). There are (up to strong isomorphism) three more simple drawings of~$K_6$ that consist of only twisted drawings of $K_5$ and they are all weakly isomorphic to generalized twisted drawings (see Figure~\ref{fig:gtwisted_k6}).
	
	\bibliography{gtwisted}

\appendix

	\section{Proof of Lemma~\ref{lem:twist}}\label{appendix:lem_twist}
	
	\twist*
	\begin{proof}
		Assume, for a contradiction, that the edge~$v_2v_4$ crosses the edge~$v_1v_3$. 
		Since any simple drawing of $K_4$ has at most one crossing, no other edges of $D$ can cross.
		Recall that in any {\gtwisted} drawing, all edges are drawn {\cmonotone} and intersect the ray $r$. 
		For every edge, this determines in which direction it emanates from its vertices.
	Hence there are (up to strong isomorphism) two possibilities how the crossing edges~$v_1v_3$ and~$v_2v_4$ can be drawn in $D$, depending on whether~$v_1v_3$ crosses the ray from $O$ through~$v_4$ at a point $x_3$ before or after $v_4$; cf.\ Figure~\ref{fig:appendix_prop1}.
	In both cases, $v_1v_2$ has to cross the ray from $O$ through $v_4$ at a point $x_2$.
	This point $x_2$ has to lie after $v_4$ in the first case and before $v_4$ in the second case.
	In both cases, as the edge $v_3v_4$ has to cross $r$, it must emanate from $v_4$ in the interior of the triangular region bounded by the segment $x_2x_3$, the portion $v_1x_3$ of $v_1v_3$, and the portion $v_1x_2$ of $v_1v_2$. 
		However, the vertex $v_3$ is in the exterior of that triangular region, and therefore $v_3v_4$ would have to cross the segment $x_2x_3$, contradicting that $D$ is {\cmonotone}, or one of $v_1v_2$ and $v_1v_3$, contradicting the simplicity of $D$. 
		\begin{figure}[htb]
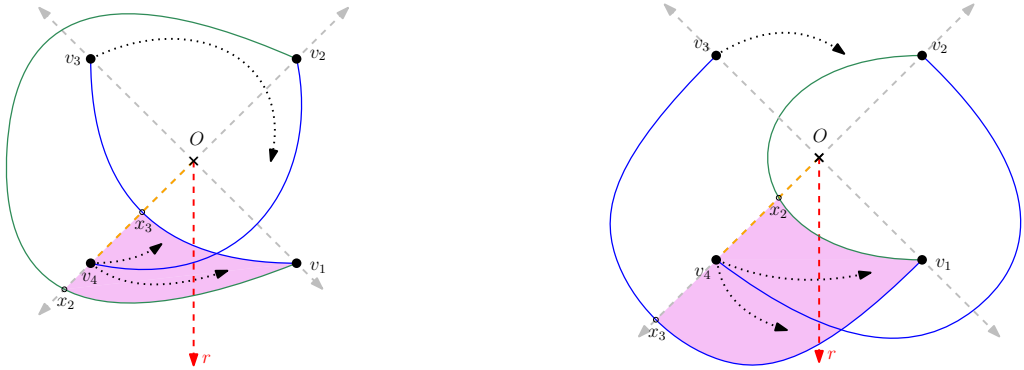

		\centering
		\begin{subfigure}[b]{0.4\textwidth}
			\centering
			\includegraphics[scale=0.6,page=5]{propertie1}
		\end{subfigure}
		\hfill
		\begin{subfigure}[b]{0.4\textwidth}
			\centering
			\includegraphics[scale=0.6,page=6]{propertie1}
		\end{subfigure}
		\caption{The two possibilities to draw $v_1v_3$ and~$v_2v_4$ crossing and \gtwisted.   
		}
		\label{fig:appendix_prop1}
	\end{figure}	
	\end{proof}
	
	\section{Proof of Lemma~\ref{lem:quasi_c_monotone}}\label{appendix:quasi_c_monotone}
	Lemma~\ref{lem:quasi_c_monotone} has been implicitly shown in \cite{bound_2014} and \cite{triangles_together}. For completeness, we include 
	a detailed proof of the lemma in this appendix. We remark that the proof presented here is in parts similar to the one in~\cite{triangles_together}.

	\begin{figure}[htb]
		\centering
		\includegraphics[page=11, scale=0.92]{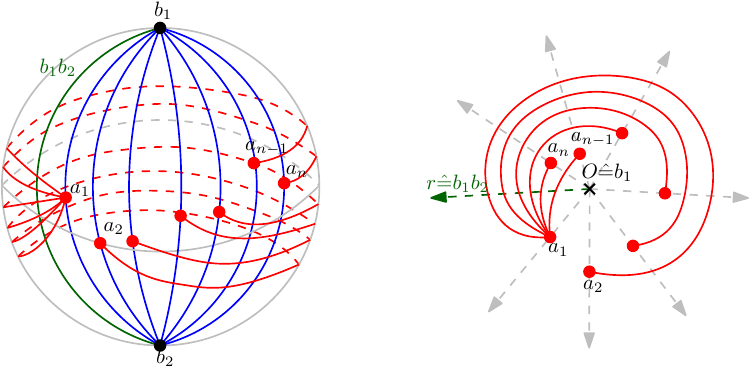}
		\caption{The homeomorphisms of $D'$. Left: $D_A$, the edges in $R_{A}$ and $b_1b_2$ are drawn on the sphere, such that $R_{A}$ and $b_1b_2$ are meridians. Right: The steographic projecton from $b_2$.}
		\label{fig:homeomorphism}
	\end{figure}
	
	\quasi*
	
	\begin{proof}
		We call the pair of edges in $D'$ incident to $a_i$, $1 \leq i \leq n$, the long edge $r_i$. Let $R_{A}$ be the set of long edges. We first show that any edge between vertices in $A$ crosses any long edge at most once. Then we show how to draw $D'$ such that $b_1$ can be taken as origin $O$ in a c-monotone drawing weakly isomorphic to $D_A$, where the long edges of $R_A$, as well as the edge $b_1b_2$, emanate as rays to infinity.
		
		We now show that every edge between two vertices of $A$ crosses every edge of $R_A$ at most once. Let $a_1$, $a_2$, and $a_3$ be vertices in $A$. 	
		Let $R_1$ be the region bounded by the edges $b_1a_1$, $a_1b_2$, $b_2a_2$ and $a_2b_1$ that does not contain $a_3$.	
		Let $R_2$ be the region bounded by the edges $b_1a_2$, $a_2b_2$, $b_2a_3$ and $a_3b_1$ that does not contain $a_1$.
		Since $D'$ is plane, these regions are disjoint.

		As the edge $e=a_1a_2$ is incident to all edges on the boundary of $R_1$, it cannot cross it. Thus, $e$ has to lie either completely inside or completely outside $R_1$ (and meet the boundary only in its endvertices). If $e$ lies inside $R_1$, it can cross neither $a_3b_1$ nor $a_3b_2$. If it lies outside $R_1$, it has to cross the boundary of $R_2$ an odd number of times. (Since $e$ must begin at $a_1$ outside $R_2$ and finish at $a_2$ inside $R_2$, and
		passing through $R_1$ is not possible.) As $e$ cannot cross edges incident to $a_2$, this means it has to cross exactly one of the edges $a_3b_1$ or $a_3b_2$. Thus, $e$ crosses the long edge $r_3$ at most once, for any vertex $a_3$.
		
		We can draw $D'$ such that $b_1$ is functioning as the origin and $R_{A}$ as rays emerging from it by doing the following transformations; see Figure~\ref{fig:homeomorphism}. We draw the subdrawing induced by the vertices of $D'$ on the sphere such that $b_1$ and $b_2$ are antipodes, and the long edges of $R_{A}$, as well as the edge $b_1b_2$, are meridians. By the general Jordan-Schoenflies theorem~\cite{homeomorphism1, homeomorphism2}, the drawing on the sphere is homeomorphic to the original drawing on the plane. We then apply a stereographic projection from $b_2$ onto the plane. This way, the long edges in $R_{A}$ and the edge $b_1b_2$ correspond to rays emerging from vertex $b_1$, where the long edges in $R_{A}$ are exactly the rays through the vertices of $D'$.

		Finally, we can obtain a c-monotone drawing that is weakly isomorphic to $D_A$. We consider the stereographic projection. As all edges of $D_A$ cross the long edges in $R_{A}$ only once, they cross in between two long edges (or rays in the projection) $r_1$ and $r_2$ if and only if their order along the rays changes (that is, the edge closer to $b_1$ at $r_1$ is further away from $b_1$ at $r_2$). Consequently, we can draw the edge-segments between every two rays as straight-lines and obtain a c-monotone drawing that is weakly isomorphic to $D_A$. If all edges of $D_A$ cross the edge $b_1b_2$, they cross a ray to infinity in the weakly isomorphic c-monotone drawing, and thus the c-monotone drawing is also generalized twisted.
	\end{proof}
	
	\section{Proof of Theorem~\ref{thm:subdrawing}}\label{appendix:thm_subdrawing}
	\thmsubdrawing*
	
	\begin{proof}
		Without loss of generality we may assume that the vertices of $D$ appear counterclockwise around $O$ in the order $v_1,v_2,\ldots ,v_n$. Let $r$ be a ray emanating from $O$, keeping $v_1$ and~$v_n$ on different sides. We define an order, $\preceq $, in this set of vertices as follows:  $v_i\preceq v_j$ if and only if either $i=j$ or $i<j$ and the edge $(v_i,v_j)$ crosses $r$.
		
		We show that $\preceq$ is a partial order. The relation is clearly reflexive and antisymmetric. Besides, if $v_i\preceq v_j$ and $v_j\preceq v_k$, then $i<j$ and $j<k$ imply $i<k$, so for the transitive property, we only have to prove that if $v_iv_j$ and $v_jv_k$ cross $r$, then $v_iv_k$ also crosses $r$. We denote by $r_i,r_j,r_k$ the rays emanating from $O$ and passing trough $v_i,v_j,v_k$, respectively. 
		We have two cases depending on where $v_jv_i$ crosses the ray $r_k$ at a point $x_k$; in the first case, $x_k$ is located before $v_k$ on $r_k$, while in the second one it is located after $v_k$. 
		Then $v_jv_k$ has to cross the ray $r_i$ at a point $x_i$, which is after $v_i$ in the first case and before $v_i$ in the second case (see Figure~\ref{fig:appendix_PartialOrder}). 
		Let $Q$ be the region bounded by the segments $Ox_i,Ox_k$ and the portions
		$v_jx_i, v_jx_k$ of the edges $v_jv_k,v_jv_i$, respectively. 
		In both cases, the edge $v_kv_i$ cannot be contained in the counterclockwise wedge from $r_i$ to $r_k$, because $v_iv_k$ should connect a vertex placed outside $Q$ with points placed inside that region, contradicting either the simplicity or the {\cmonotonicity} of $D$. 
		Therefore, $v_iv_k$ must be in the clockwise wedge from $r_i$ to $r_k$ and thus crosses the ray $r$.
		
		In this partial order $\preceq $, a chain consists of a subset $v_{i_1},\ldots ,v_{i_{s-1}}$ of pairwise comparable vertices, that is, a subset of vertices such that their induced subdrawing is {\gtwisted} (all edges cross $r$). An antichain, $v_{j_1},\ldots ,v_{j_{t-1}}$, consists of a subset of pairwise incomparable vertices, that is, a subset of vertices such that their induced subdrawing is monotone (no edge crosses $r$). 
		Therefore, the first part of the theorem follows from applying Theorem~\ref{thm:dilworth} to the set of vertices of $D$ and the partial order $\preceq $.
		
		Finally, observe that if $s=t\leq \lceil \sqrt{n}\rceil $, then $(s-1)(t-1)+1\leq n$. Thus, $D$ contains a complete subgraph $K_{\lceil \sqrt{n}\rceil}$ whose induced subdrawing is either  {\gtwisted} or monotone.
			\begin{figure}
			\begin{subfigure}[b]{0.45\textwidth}
				\centering
				\includegraphics[scale=0.6, page=9]{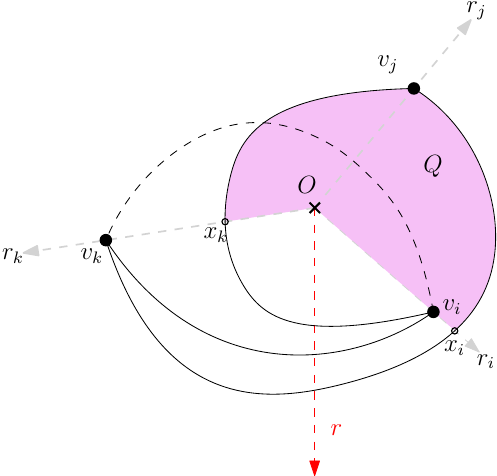}
			\end{subfigure}
			\hfill
			\begin{subfigure}[b]{0.45\textwidth}
				\centering
				\includegraphics[scale=0.6, page=10]{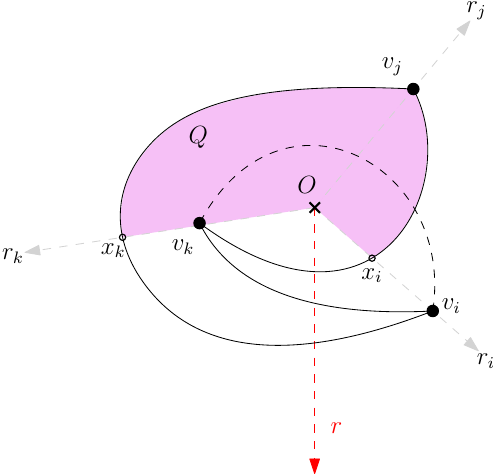}
			\end{subfigure}
			\caption{If edges $v_iv_j$ and $v_jv_k$ cross $r$ in a c-monotone drawing, then $v_iv_k$ must also cross~$r$.}\label{fig:appendix_PartialOrder}
		\end{figure}
	\end{proof} 
	
	\section{{\Gtwisted} drawings contain a pair of antipodal vi-cells}\label{appendix:vicell}
	
	In this section, we will show that every drawing weakly isomorphic to a generalized drawing of $K_n$ contains a pair of antipodal vi-cells (Theorem~\ref{thm:vicells}). Before proving the theorem, we will see some useful properties of {\gtwisted} drawings. Recall that in a {\gtwisted} drawing, vertices are labeled $v_1, v_2, \ldots , v_n$ counterclockwise around the origin $O$, the ray emanating from $O$ and passing through a vertex $v_i$ is denoted by $r_i$, and the ray $r$ that emanates from $O$ and crosses every edge once is between $r_n$ and $r_1$, counterclockwise from~$r_n$.
	
	\begin{figure}[htb]
		\centering
		\includegraphics[page=35, scale=0.6]{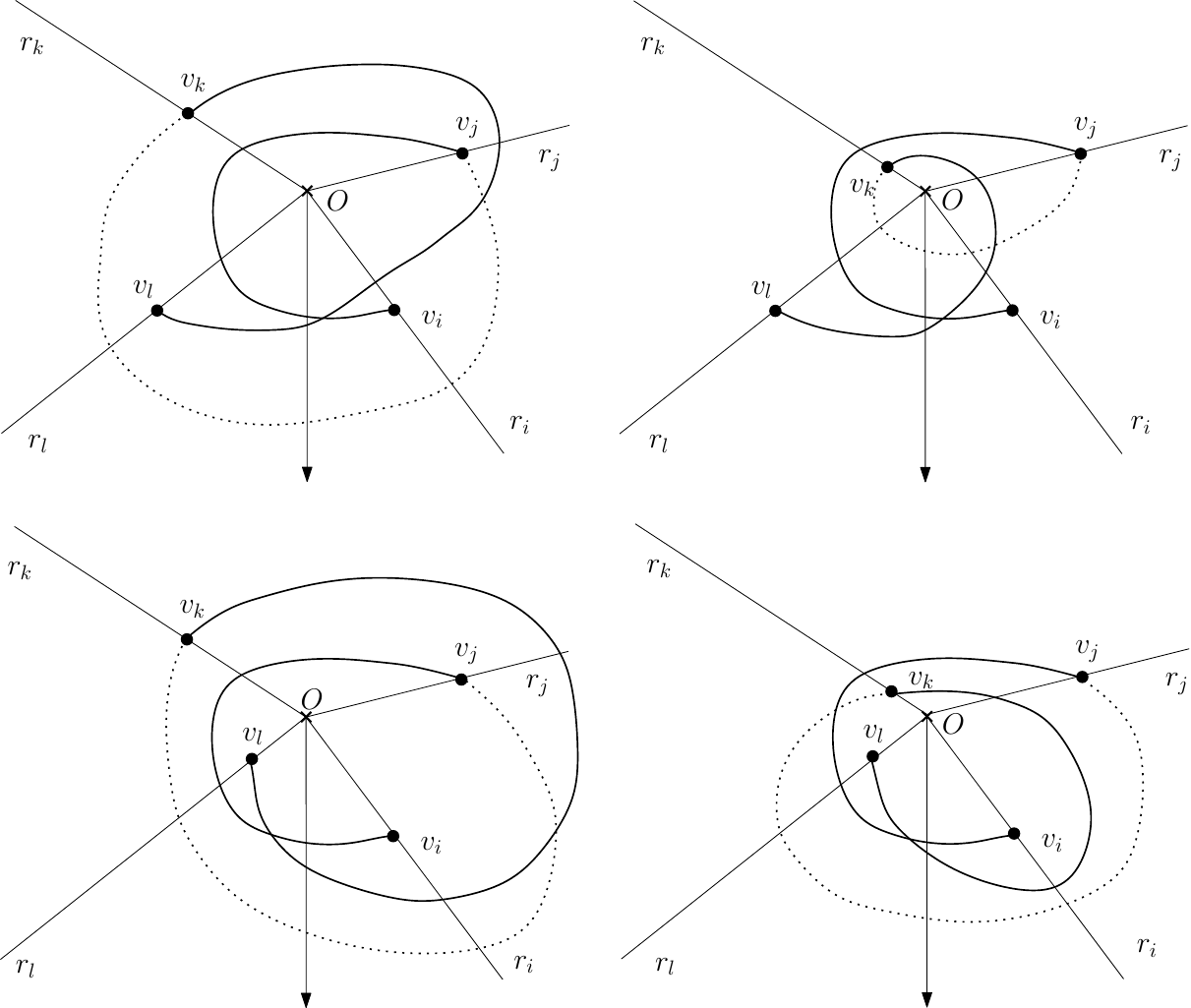}
		\caption{Illustrating the proof of Lemma~\ref{lem:cross}.}\label{fig:cross}
	\end{figure}

	\begin{restatable}{lem}{lemcross}\label{lem:cross}
		Let $D$ be a {\gtwisted} drawing of $K_n$ with $n\ge 4$. Suppose the two edges $v_iv_j$ and $v_kv_l$ of $D$ cross, and $i<j<k<l$. Then the crossing point between these two edges is in the wedge $W$ defined by $r_j$ and $r_k$, counterclockwise from $r_j$ to $r_k$.
	\end{restatable}

	\begin{proof}
		Assume for contradiction that the crossing point is not in $W$, so it is in the wedge defined by $r_l$ and $r_i$, counterclockwise from $r_l$. There are four cases, depending on whether $v_k$ and $v_l$ are to the left or the right of the directed edge $v_jv_i$; see Figure~\ref{fig:cross}. In any of the four cases, there is no way of connecting $v_k$ and $v_j$ without crossing either $v_iv_j$ or $v_kv_l$, which is a contradiction.
	\end{proof}

	\begin{lem}\label{lem:empty}
		For every {\gtwisted} drawing $D$ with $n\ge 3$ vertices, the following statements hold.
		\begin{enumerate}
			\item[i)]\label{lem:emptyi} There exists a vertex $v_i$, with $1\le i \le n-1$, such that the bounded region $RB$ defined by the edge $v_iv_{i+1}$ and the segments $\overline{Ov_i}$ and $\overline{Ov_{i+1}}$ is empty.
			\item[ii)]\label{lem:emptyii}  There exists a vertex $v_j\ne v_i$, with $1\le j \le n-1$, such that the unbounded region $RU$ defined by the edge $v_jv_{j+1}$ and the segments $\overline{Ov_j}$ and $\overline{Ov_{j+1}}$ is empty.
		\end{enumerate}
	\end{lem}
	
	\begin{proof}
		We show statement \emph{i)}, we take the first edge $v_iv_k$ (with $i < k$) that crosses $r$ and we show that $v_i$ satisfies \emph{i)}. Let $x$ be the crossing point between $v_iv_k$ and $r$, and let $R$ be the bounded region defined by the edge $v_iv_{k}$ and the segments $\overline{Ov_i}$ and $\overline{Ov_k}$.
		
		\begin{figure}[htb]
			\centering
			\includegraphics[page=36, scale=0.7]{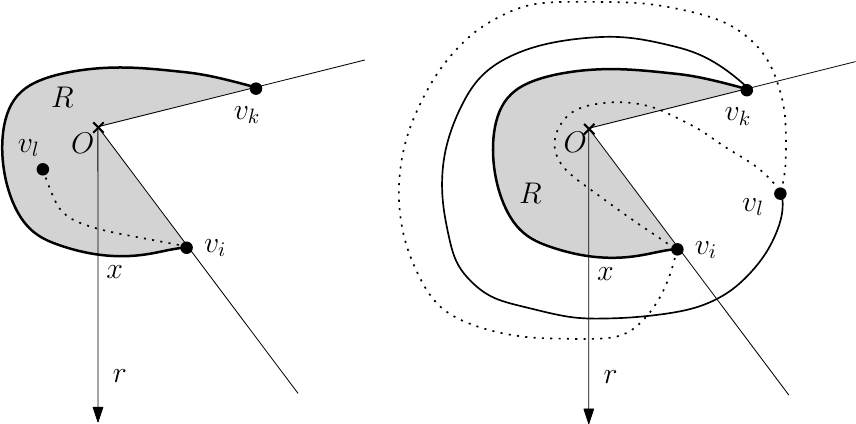}
			\caption{Illustrating the proof of Lemma~\ref{lem:empty}.}\label{fig:empty}
		\end{figure}
		
		Suppose that there is a vertex $v_l$ inside $R$. See Figure~\ref{fig:empty}, left. Then there is no way of connecting $v_l$ and $v_i$ without crossing $r$ before $x$, which contradicts that $v_iv_{k}$ is the first edge crossing $r$. Thus, $R$ must be empty.
		
		Suppose now that $k\ne i+1$, so there is a vertex $v_l$ with $i < l < k$. See Figure~\ref{fig:empty}, right. The edge $v_lv_k$ must cross $r$ at a point after $x$. But then there is no way of adding the edge $v_iv_l$ without crossing $r$ before $x$ or without crossing $v_lv_k$. Therefore, $k=i+1$ and \emph{i)} follows.
		The proof of \emph{ii)} is analogous by taking the last edge crossing $r$. In addition, $v_i$ and $v_j$ must be different since a same edge $v_iv_{i+1}$ cannot be at the same time the first and the last edge crossing $r$.
	\end{proof}

	\begin{lem}\label{lem:vi}
		Let $D$ be a {\gtwisted} drawing of $K_n$ with $n\ge 3$ vertices. Then the cell containing $O$ and the unbounded cell have at least one vertex on their boundaries.
	\end{lem}
	
	\begin{proof}
		We show that the cell containing $O$ is a vi-cell. The proof for the unbounded cell follows analogously.
		
		By Lemma~\ref{lem:empty} \emph{i)}, there exists a vertex $v_i$, with $1\le i \le n-1$, such that the bounded region $RB$ defined by the edge $v_iv_{i+1}$ and the segments $\overline{Ov_i}$ and $\overline{Ov_{i+1}}$ is empty. We now show that either the segment $\overline{Ov_i}$ or the segment $\overline{Ov_{i+1}}$ is uncrossed. Thus, it follows immediately that $O$ lies in a vi-cell (with either $v_i$ or $v_{i+1}$ on the boundary).

		\begin{figure}[ht]
			\centering
			\begin{subfigure}{.42\textwidth}
				\includegraphics[scale=0.6,page=31]{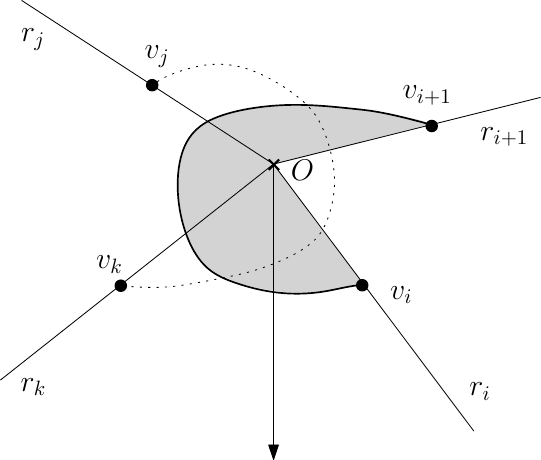}
				\caption{}\label{fig:the1}
			\end{subfigure}
			\hfill
			\begin{subfigure}{.42\textwidth}
				\includegraphics[scale=0.6,page=32]{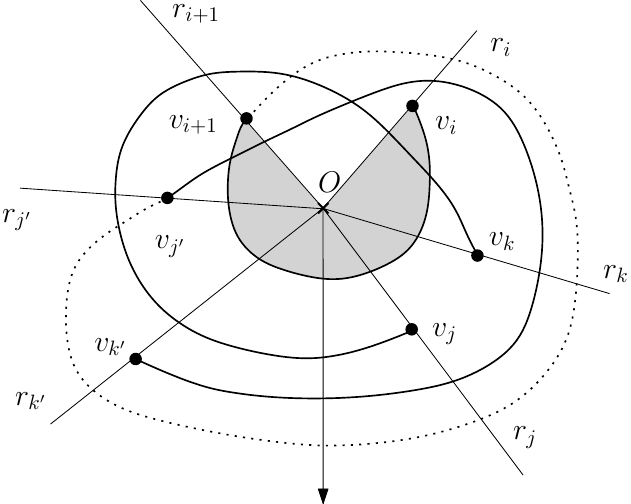}
				\caption{}\label{fig:the2}
			\end{subfigure}
			
			\begin{subfigure}{.42\textwidth}
				\centering
				\includegraphics[scale=0.6,page=33]{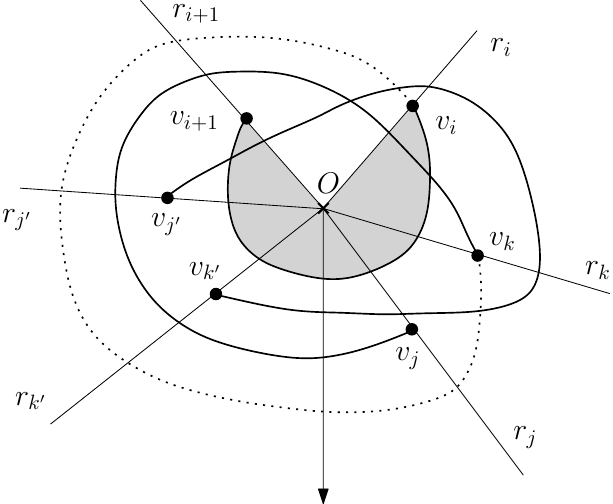}
				\caption{}\label{fig:the3}
			\end{subfigure}
			\caption{One of the segments $\overline{Ov_i}$ and  $\overline{Ov_{i+1}}$ is uncrossed.}\label{fig:the}
		\end{figure}
		
		Suppose to the contrary that neither segments are uncrossed, so there is an edge $v_jv_k$, with $j<k$, crossing the segment $\overline{Ov_i}$, and another edge $v_{j'}v_{k'}$, with $j' < k'$, crossing the segment $\overline{Ov_{i+1}}$.
		
		Observe now the following. First, since $RB$ is empty, an edge $v_{i+1}v_l$ cannot cross $\overline{Ov_i}$ for any $l$, so neither $v_j$ nor $v_k$ can be $v_{i+1}$. Second, suppose that $j,k > i+1$. Since both vertices are outside $RB$, if $v_jv_k$ crosses $\overline{Ov_i}$, then it must also cross $v_iv_{i+1}$ (see Figure~\ref{fig:the1}). Hence, by Lemma~\ref{lem:cross}, the crossing point between $v_jv_k$ and $v_iv_{i+1}$ is in the wedge defined by $r_{i+1}$ and~$r_j$. But then, after crossing $v_iv_{i+1}$ and~$\overline{Ov_i}$, the edge $v_jv_k$ must cross $v_iv_{i+1}$ a second time to reach~$v_k$, which is a contradiction. Therefore, $j,k < i$.
		
		Using an analogous reasoning, we also obtain that if an edge $v_{j'}v_{k'}$ crosses~$\overline{Ov_{i+1}}$, then~$j',k' > i+1$. As a consequence, the relative position of these three edges of $D$ is as shown in Figures~\ref{fig:the2} and~\ref{fig:the3}. After emanating in~$v_k$, the edge $v_kv_j$ crosses $v_iv_{i+1}$ at a point in the wedge defined by $r_k$ and~$r_i$; then it crosses~$\overline{Ov_i}$, and reaches $v_j$ surrounding $v_iv_{i+1}$ by the exterior. The edge $v_{j'}v_{k'}$ must do the same but in opposite direction, crossing first $v_iv_{i+1}$ at a point in the wedge defined by $r_{i+1}$ and~$r_{j'}$, then crossing $\overline{Ov_{i+1}}$ and reaching $v_{k'}$ surrounding $v_iv_{i+1}$ by the exterior. Note that $v_jv_k$ and $v_{j'}v_{k'}$ necessarily cross, and that the crossing point must be in the wedge defined by $r_i$ and~$r_{i+1}$. Hence, $v_{j'}$ must be to the left of the oriented edge $v_kv_j$ and $v_k$ must be to the right of the oriented edge~$v_{j'}v_{k'}$. Otherwise, if $v_{j'}$ is to the right of the oriented edge $v_kv_j$ or $v_k$ is to the left of the oriented edge~$v_{j'}v_{k'}$, then edges $v_jv_k$ and $v_{j'}v_{k'}$ would cross twice.
		
		For the vertices $v_j$ and~$v_{k'}$, there are two possibilities, depending on whether $v_{k'}$ is to the right (see Figure~\ref{fig:the2}) or to the left (see Figure~\ref{fig:the3}) of the oriented edge $v_kv_j$. But in the first case, the edge $v_{j'}v_{i+1}$ would cross $v_jv_k$ twice, and in the second case, the edge $v_iv_k$ would cross $v_{j'}v_{k'}$ twice. Therefore, at least one of $\overline{Ov_i}$ and $\overline{Ov_{i+1}}$ is uncrossed.
	\end{proof}

	\begin{lem}\label{lem:gtwisted_vicells}
		Let $D$ be a {\gtwisted} drawing of $K_n$ with $n\ge 3$ vertices. Then the cell of $D$ containing $O$ an the unbounded cell are a pair of antipodal vi-cells.
	\end{lem}
	
	\begin{proof}
		Let $c$ be the segment $\overline{OZ}$, where $O$ is the origin and $Z$ is a point on $r$ in the unbounded cell. By Lemma~\ref{lem:vi}, both $O$ and $Z$ lie in vi-cells. We will show that those cells are antipodal. Since $r$ crosses every edge of $D$ exactly once and $Z$ lies in the unbounded cell, also the segment $c$ crosses every edge exactly once. Consequently, $c$ crosses the boundary of every triangle of $D$ exactly three times. Since every triangle of $D$ is plane, this means $c$ starts and ends at different sides of every triangle. Thus, $O$ and $Z$ have to lie in antipodal cells.
	\end{proof}
	
\begin{figure}[h!tb]
		\centering
		\includegraphics[page=2,scale=0.5]{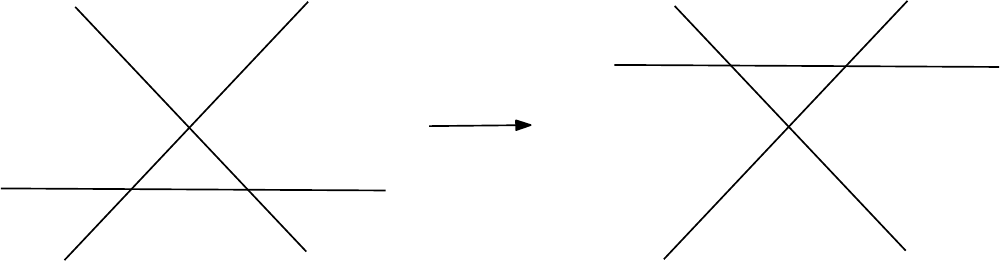}
		\caption{A triangle-flip.}
		\label{fig:triangle_flip}
\end{figure}

	We extend Lemma~\ref{lem:gtwisted_vicells} to drawings weakly isomorphic to generalized twisted drawings using Gioan's Theorem~\cite{arroyo,gioan} that any two weakly isomorphic drawings of $K_n$ can be transformed into each other with a sequence of triangle-flips and at most one reflection of the drawing. A \emph{triangle-flip} is the operation that transforms a triangular cell $\triangle$ that has no vertex on its boundary, by moving one of its edges across the intersection of the two other edges of~$\triangle$ (see Figure~\ref{fig:triangle_flip}).
	
	\thmvicells*
	
	\begin{proof}
		Let $D$ be a simple drawing of $K_n$ that is weakly isomorphic to a generalized twisted drawing $D'$. By Lemma~\ref{lem:gtwisted_vicells}, the cell in which $O$ lies in $D$ and the unbounded cell of $D$ are antipodal vi-cell pairs. 
		Without loss of generality, we can assume that~$O$ is very close to a vertex~$v$ on the boundary of~$C_1$ and $Z$ is very close to a vertex~$w$ on the boundary of~$C_2$.
		Using Gioan's Theorem, it is enough to show that afer every triangle flip that can be transformed on a drawing $\tilde{D}$ containing antipodal vi-cells, the resulting drawing $\tilde{D}_2$ still contains antipodal vi-cells.
		
		As triangle-flips are only applied to cells without vertices on its boundary, and points $O$ and $Z$ are close enough to vertices on the boundary of their cell in~$\tilde{D}$, they stay in vi-cells (with the vertices they are close to) after every triangle-flip. What remains to be shown is that the vi-cells stay antipodal.
		
		Let $T$ be a triangle of~$\tilde{D}$. (Note that a triangle of the drawing is the simple cycle formed by the three edges connecting three vertices of the graph, and not the triangular cells on which we perform triangle flips.) Whenever a flip is performed on a cell $\triangle$, the cell $\triangle$ disappears and a new cell appears. The new cell might (but not has to) be on the other side of $T$, but no other cells are affected. In particular, since triangle flips are never applied on vi-cells, any pair of vi-cells that was antipodal before stays antipodal after the flip.
	\end{proof}
	
	\section{Characterizing via antipodal vi-cells}\label{appendix:antipodal_gtwisted}
	
	In this section, we prove that in any simple drawing $D$ of $K_n$ containing antipodal vi-cells it is possible to add a simple curve crossing every edge of $D$ exactly once (Theorem~\ref{thm:antipodal_then_gtwisted}). To this end, we will use the following lemmata that show some properties of antipodal vi-cells.

	\begin{lem}\label{prop:different_v}
		Let $D$ be a drawing of $K_n$ with $n \geq 4$, and let $(C_1,C_2)$ be a pair of antipodal vi-cells. Then there is no vertex that lies on the boundary of both cells.
	\end{lem}
	
	\begin{proof}
		Assume, for a contradiction, that a vertex $v_1$ lies on the boundary of both cells $C_1$ and $C_2$. See Figure~\ref{fig:not_same}. Consider another cell $C'$ different from $C_1$ and $C_2$, with $v_1$ on its boundary. Since $v_1$ has degree $n-1 \geq 3$, this cell exists. Let $u_1$ and $u_2$ be the two vertices whose edges to $v_1$ are on the boundary of that cell $C'$. Then the triangle formed by vertices $u_1$, $u_2$, and $v_1$ always have the cells $C_1$ and $C_2$ on the same side, contradicting that $C_1$ and $C_2$ are antipodal.
		\begin{figure}[h!tb]
			\centering
			\includegraphics[page=7, scale=0.91]{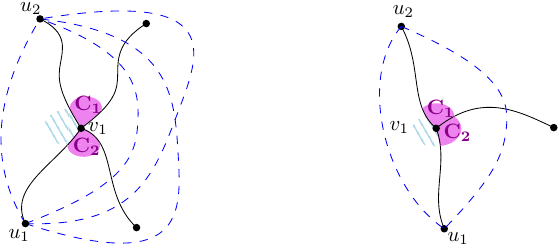}
			\caption{Edges of the star $S(v_1)$ incident to $v_1$ are drawn black; the antipodal cells are filled purple; the additional cell is indicated in cyan; the different ways to draw the edge $u_1u_2$ are drawn dashed in blue.}\label{fig:not_same}
		\end{figure}	
	\end{proof}
	
	\begin{lem}\label{lem:no_points_inside}
		Let $D$ be a simple drawing of $K_n$ that contains two antipodal vi-cells $C_1$ and~$C_2$. Let $v_1$ be a vertex on the boundary of $C_1$. Let $T$ be a triangle formed by $v_1$ and two other vertices $u_2$ and $u_3$. If there is a vertex $u$ that lies on the same side of $T$ as $C_1$, then the edge $v_1u$ lies completely on that side.
	\end{lem}
	
	\begin{proof}
		Assume, for a contradiction, that the edge $v_1u$ does not lie completely on the same side of $T$ as $C_1$, and thus crosses the boundary of $T$. See Figure~\ref{fig:inside}. The only edge it can cross is $u_2u_3$, as the other edges are incident to $v_1$. Thus, the drawing induced by $v_1$, $u_2$, $u_3$ and $u$ contains a crossing between $v_1u$ and $u_2u_3$. Since any simple drawing of $K_4$ contains at most one crossing, the edges $uu_2$ and $uu_3$ cannot cross the boundary of $T$. Thus, the triangle $T'$ formed by $u$, $u_2$ and $u_3$ has to lie on the same side of $T$ as $C_1$, but keeping $C_1$ and $C_2$ on one of its sides, which is a contradiction to the definition of antipodal.
	\end{proof}

	\begin{figure}[htb]
		\centering
		\includegraphics[page=1, scale=0.91]{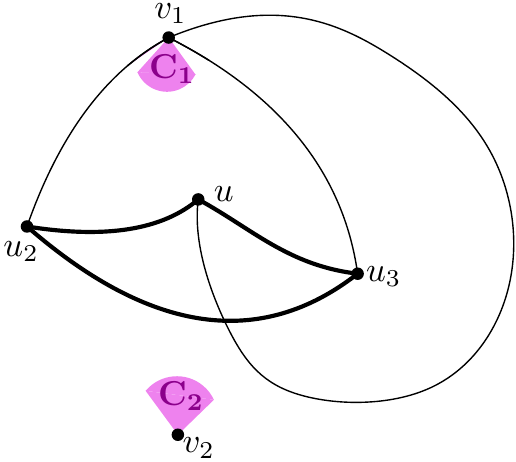}
		\caption{An illustration of Lemma~\ref{lem:no_points_inside}. The vertex $v_2$ is placed as an example in one of the possible faces and could be at another place, but $C_1$ and $C_2$ lie on different sides of the triangle $v_1u_2u_3$ by definition. Thus, by construction of the triangle $uu_2u_3$, the cells $C_2$ and $C_1$ lie on the same side of $uu_2u_3$.}
		\label{fig:inside}
	\end{figure}

	Now, we can prove Theorem~\ref{thm:antipodal_then_gtwisted}.
	
	\thmantipodal*

	\begin{proof}
		Let $(C_1,C_2)$ be a pair of antipodal vi-cells of $D$,
		$v_1$ a vertex on the boundary of $C_1$, $v_2$ a vertex on the boundary of $C_2$, and $S(v_1)$ the star of $v_1$. Note that by Lemma~\ref{prop:different_v}, $v_1$ and $v_2$ are different. We draw a simple curve $c$ from $v_1$ to $v_2$ such that it emerges from $v_1$ in the cell $C_1$ and ends in the cell $C_2$ very close to $v_2$, and the following holds:

		\begin{enumerate}
			\item The curve $c$ does not cross any edges of~$S(v_1)$.
			\item All intersections of $c$ with edges of $D$ are proper crossings.
			\item Over all curves for which 1 and 2 hold, the curve $c$ has the minimum number of crossings with edges of $D$.
		\end{enumerate}
		Since $S(v_1)$ is a plane drawing that has only one face in which both $v_1$ and $v_2$ lie, drawing $c$ is always possible. See for example Figure~\ref{fig:curve}. We will prove that $c$ crosses every edge of $D \setminus S(v_1)$ exactly once.
		To show that $c$ crosses an arbitrary edge $w_2w_3$ exactly once, we will first show that $c$ crosses $w_2w_3$ an odd number of times and then show that $c$ crosses $w_2w_3$ at most twice.
		
		\begin{figure}[htb]
			\centering
			\includegraphics[page=5, scale=0.91]{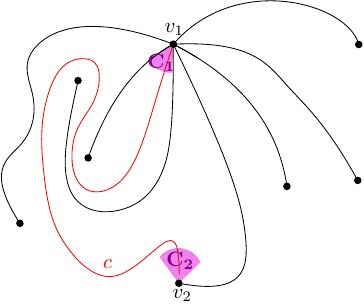}
			\caption{Edges of the star $S(v_1)$ incident to $v_1$ are drawn black; the antipodal cells are filled purple; the curve $c$ is drawn in red.}\label{fig:curve}
		\end{figure}
		
		To observe that $c$ has to cross $w_2w_3$ an odd number of times, consider the triangle $T$ formed by $v_1$, $w_2$ and $w_3$. Since $c$ connects two antipodal cells, the endpoints of $c$ have to be on two different sides of $T$. Thus, $c$ has to cross the boundary of $T$ an odd number of times. Since $c$ does not cross $S(v_1)$, it has to cross $w_2w_3$ an odd number of times.
		
		We show now that $c$ crosses $w_2w_3$ at most twice. Assume to the contrary that $c$ crosses $w_2w_3$ at least three times. Without less of generality, we may assume that $C_1$ is inside $T$ and $C_2$ is outside. Given two crossing points $x$ and $y$ between $c$ and $w_2w_3$ that are consecutive on $c$ when going from $v_1$ to $v_2$, a \emph{lens} is the region to the left of the cycle formed by the arc $xy$ on $c$ and the arc $yx$ on $w_2w_3$. See Figure~\ref{fig:lens_minimal} for an illustration. The pairs of two consecutive crossing points between $c$ and $w_2w_3$ on $c$ define a set of lenses on both sides of $T$, possibly nested (see Figure~\ref{fig:lens_minimal}). Note that since $c$ crosses $w_2w_3$ at least $3$ times, there is at least one lens on each side of $T$. Among all the lens on the same side of $T$ as $C_1$, we take one that does not contain any other lens in its interior. This lens $L$ always exists by taking the "innermost" one in a set of nested lenses.

		\begin{figure}[htb]
			\centering
			\includegraphics[page=3, scale=0.91]{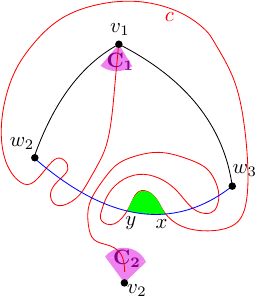}
			\caption{The curve $c$ crossing several times the edge $w_2w_3$. If there are nested lenses in $T$, we take a minimal one that does not contain any nested lenses. This minimal lens $L$ is shaded green.}\label{fig:lens_minimal}
		\end{figure}
		
		Let $x$ and $y$ be the two crossing points defining the lens, so the boundary of $L$ consists of arc $xy$ on $c$ and arc $yx$ on $w_2w_3$. We claim that there is no vertex in $L$. Assume that there is a vertex $u$ inside $L$ (see Figure~\ref{fig:lens_with_vertex}). By Lemma~\ref{lem:no_points_inside}, the edge $uv_1$ cannot cross the edge $w_2w_3$. Thus, it has to cross $c$ in order to get from $u$ inside the lens to $v_1$ outside the lens. This is a contradiction to $c$ being drawn such that it does not cross $S(v_1)$. Thus, there is no vertex in $L$, as claimed.
		
		\begin{figure}[htb]
			\centering
			\includegraphics[page=2, scale=0.91]{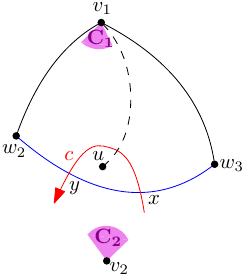}
			\caption{There is no vertex in $L$.
				If there is a vertex $u$ inside $L$, then the edge from $u$ to $v_1$ (drawn dashed) would have to cross $c$, which it cannot by construction.}
			\label{fig:lens_with_vertex}
		\end{figure}
		
		Since $L$ does not contain any vertex, then every edge that crosses the arc $yx$ on $w_2w_3$ has to also cross at least once the arc $xy$ on the curve $c$ (as there is no vertex in the lens where it could stop and it cannot cross the edge $w_2w_3$ more than once); see Figure~\ref{fig:lens_empty}~(left). Thus, $c$ can be drawn such that it stops before $x$, follows an arc very close to the arc $xy$ on $w_2w_3$ until a point very close to $y$, and then continues as $c$ did before; see Figure~\ref{fig:lens_empty}~(right). This way the new drawing of the curve is still simple and does not have any crossings that the original one did not, but two less crossings with $w_2w_3$, which is a contradiction to the minimality of~$c$. In conclusion, the curve $c$ cannot cross any edge $w_2w_3$ more than twice.
		
		\begin{figure}[htb]
			\centering
			\includegraphics[page=4, scale=0.91]{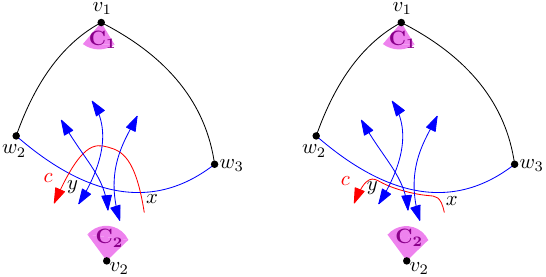}
			\caption{Left: The lens $L$ does not contain any vertices, thus all edges crossing $w_2w_3$ within the lens have to leave the lens crossing~$c$. Right: The curve $c$ is redrawn such that it has fewer crossings.}
			\label{fig:lens_empty}
		\end{figure}
		
		Therefore, since $c$ crosses all edges in $D \setminus S(v_1)$ an odd number of times and at most twice, it follows that it crosses all edges in $D \setminus S(v_1)$ exactly once,  while (per construction) it does not cross any edges of~$S(v_1)$.
		
		We can transform $c$ to a curve that crosses all edges exactly once in the following way: Instead of the starting point being $v_1$ we remove an $\varepsilon$ of the curve on this end such that it starts very close to $v_1$ in the cell $C_1$ (and consequently still crosses exactly the edges it crossed before). Then, on that start in $C_1$, we extend the curve by going around the vertex $v_1$ so close to $v_1$ that the extension crosses exactly the edges of $S(v_1)$, and then ending again in $C_1$; see Figure~\ref{fig:final_curve}. This way the extension crosses all edges of $S(v_1)$ exactly once and consequently, we obtained a curve crossing all edges exactly one, with its endpoints not lying on any edges or vertices of $D$.
		\begin{figure}[htb]
			\centering
			\includegraphics[page=6, scale=0.91]{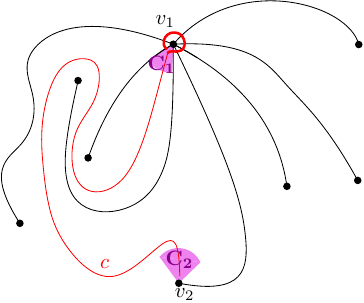}
			\caption{The resulting curve after the extension. The last part crossing $S(v_1)$ is drawn bold.}\label{fig:final_curve}
		\end{figure}
	\end{proof}

\section{Characterizing via a Curve Crossing Everything}\label{appendix:char_curve}

In this section we prove Theorem~\ref{thm:curve_then_gtwisted}.

	\thmcurvegtiwsted*

We first show several properties for drawings such that there is a simple curve $c=OZ$ crossing every edge once. Then, we show that we can extend the drawing~$D=D_n$ to a drawing~$D_{n+2}$ of $K_{n+2}$ by adding~$O$ and~$Z$ as vertices, the curve~$c$ as an edge, and edges from~$O$ and~$Z$ to each vertex~$w$ of~$D_n$ in such a way that~$D_{n+2}$ fulfills the following properties.
	\begin{enumerate}[leftmargin=*,label={\textbf{(P\arabic*)}}]
		\item\label{p:simple} $D_{n+2}$ is a simple drawing. (This implies in particular that none of the curves incident to~$O$ crosses another curve incident to~$O$ and no curve incident to~$Z$ crosses another curve incident to~$Z$.)
		\item\label{p:OZplane}  No edge incident to~$O$ crosses any edge incident to~$Z$. 
	\end{enumerate}

	\subsection*{Notation and Basic Properties}
	
	Assume that $D_n$ is a simple drawing of $K_n$ and $c=OZ$ a simple curve crossing every edge of $D_n$ once.
	
	We will consider two orientations for each edge of~$D_n$. By $uv$ we denote the edge oriented from $u$ to~$v$, and by $vu$ the same edge oriented from $v$ to~$u$. If $x,y$ are points placed in that order on the edge $uv$ of $D_{n+2}$, then the portion of the curve $uv$ placed between $x$ and $y$ is called the \emph{arc} $xy$. We also consider the arcs oriented from the first point to the second point. Consequently, $yx$ has the same points as $xy$ but with the opposite orientation. We orient $c$ from $O$ to~$Z$. When considering the star~$S(w)$ of a vertex $w$, we will always consider the edges oriented from $w$ to the other endpoint.

	When the edge $wu$ crosses $OZ$ in a crossing point $x$ we know that this crossing can be of two different ways, depending on the radial order of the arcs $xO, xw, xZ, xu$ around point $x$. We will say that $wu$ is a  \emph{top edge} if around $x$ the arcs $xu,xO,xw,xZ$ appear clockwise in this order, and $wu$ is a \emph{bottom edge} when that clockwise order is  $xu,xZ,xw,xO$. In the figures we draw the curve $c=OZ$ as a horizontal line, thus the directed edges reaching that line by its top side are precisely the top edges. Note that if $wu$ is a top edge, then $uw$ is a bottom edge.

		\begin{figure}[!htb]
			\centering
			\includegraphics[scale=0.75,page=1]{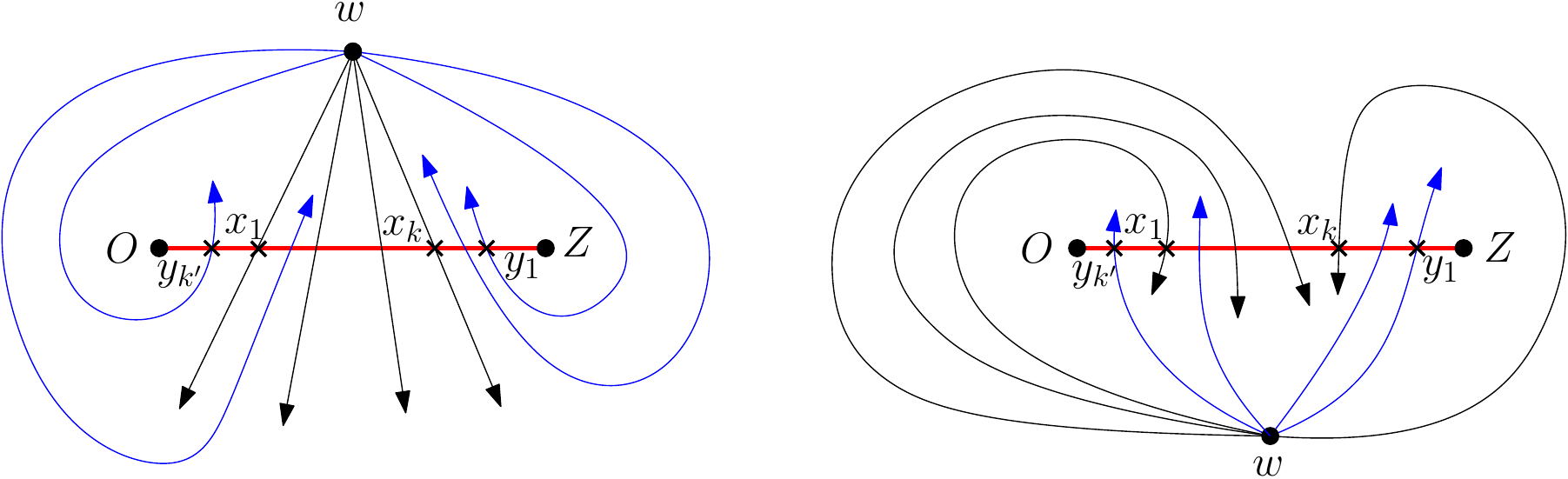}
			\caption{Top edges are drawn in black, bottom edges are drawn in blue.}\label{fig:TopBottom}
		\end{figure}

	Three arcs, $xy$ on the edge $e_1$, $yz$ on the edge $e_2$ and $zx$ on the edge $e_3$, form a cycle and divide the plane (or the sphere) into two regions $A,B$. By triangular region $xyz$ we mean the region ($A$ or $B$) found on the left side when we walk the cycle in the order $x$, then $y$, then $z$, and returning to $x$, using the corresponding arcs in~$e_1,e_2,e_3$. In the same way, if the arcs $x_1x_2,x_2x_3,\ldots ,x_{k-1}x_k,x_kx_1$ form a simple cycle, the region found on the left side when we walk the cycle in the order $x_1,x_2,\ldots ,x_k,x_1$ will be denoted by $x_1x_2\ldots x_k$.  We suppose that the drawings are on the sphere $S^2$, so we consider drawings homeomorphic in  $S^2$ as topologically identical, like the left and right drawings of Figure~\ref{fig:TopBottom}. However, in the figures, as the drawings are shown on the plane, a given region can be bounded or unbounded. For example, region $wx_1y_1$ is bounded in the left drawing of Figure~\ref{fig:TopBottom} and unbounded in the right drawing. But in both drawings, the arcs, vertices and crossings inside region $wx_1y_1$ are the same.
	
	In our constructions, we are going to draw new arcs that are close to (or glued to) arcs of~$D_n$. A new arc $a'$ is \emph{close} or \emph{glued} to the arc $ux$ if
	\begin{enumerate}[label={\textbf{(g\arabic*)}}, leftmargin=*]
		\item\label{closei} An edge crosses $a'$ if and only if it crosses $ux$.
		\item\label{closeii} All the crossing points on $a'$ and on $ux$ have the same order.
		\item\label{closeiii} The arcs $a'$ and $ux$ do not cross each other. (They can share the endpoints, points $u,x$, but do not have to.)
	\end{enumerate}

	\begin{lemma}\label{prop:propertiesa}
		Let $w$ be a vertex of $D_n$. The edges of $S(w)$ satisfy the following properties:
		
		\begin{enumerate}[label={\textbf{(a$\mathbf{_\arabic*}$)}}, leftmargin=*]
			\item\label{lem:a1} 
			When exploring counterclockwise around $w$ the edges of $S(w)$, the top edges are consecutive, $wu_1, wu_2 \ldots, wu_k$, and cross the curve $c=OZ$ at points $x_1,\ldots ,x_k$ in that order. Then the bottom edges are consecutive, $wv_1,\ldots ,wv_{k'}$ (where $k'=n-k-1$), and cross the curve $c'=ZO$ at points $y_1,\ldots ,y_{k'}$ in that order. See Figure~\ref{fig:TopBottom}.
			\item\label{lem:a2}	Let $z_1$ and $z_{n-1}$ be the first and the last crossing points of $S(w)$ on the curve $c$. Then the endpoints of the bottom edges of $S(w)$ are inside the triangular region $z_1z_{n-1}w$, and the endpoints of the top edges are outside that region.	See Figure~\ref{fig:Propertya2}.
		\end{enumerate}
	\end{lemma}
	\begin{proof}
		\ref{lem:a1} Draw a new arc $a$ glued to $OZ$ by its top part, and another arc $a'$ glued to $ZO$ on the bottom part. Both have endpoints $O$ and $Z$, and thus $a,a'$ define a cycle~$C$. The edges of the star $S(w)$ in counterclockwise order have to reach (that is, have their first crossing point with) the cycle $C$ at points placed in clockwise order on $C$. The top edges of $S(w)$ are the ones reaching $C$ on the arc~$a$, and the corresponding crossing points $x_1, \ldots, x_k$ on $OZ$ are in increasing order (from $O$ to $Z$). Then come the bottom edges reaching $C$ on the arc~$a'$, their corresponding crossing points $y_1, \ldots, y_{k'}$ in this order on $ZO$.

		\begin{figure}[!htb]
		\centering
		\includegraphics[scale=0.7,page=2]{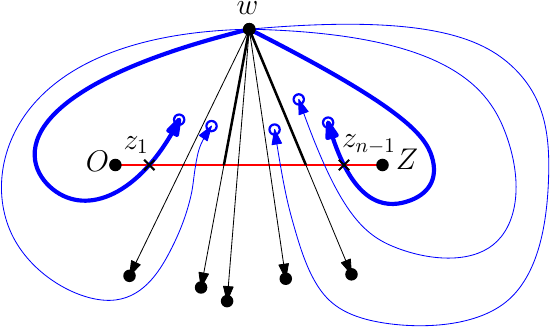}
		\caption{Property $a_2$: The vertices placed in the regions with a blue circle can be reached from $w$ only with bottom edges.}\label{fig:Propertya2}
	\end{figure}
	
		\ref{lem:a2} Notice that $z_1$ can be either $x_1$ or $y_{k'}$. Similarly, $z_{n-1}$ can be $x_k$ or $y_1$. In any case, the only way to connect $w$ to a vertex  $v$ placed inside the triangular region $z_1z_{n-1}w$ and crossing $OZ$, is crossing the arc $z_1z_{n-1}$ from its bottom part. For the same reason, all the vertices placed outside that region are precisely the endpoints of the top edges of~$S(w)$. See Figure~\ref{fig:Propertya2}.
	\end{proof}

	\begin{lemma}\label{prop:propertiesb}
		Assume $wu_1,\ldots ,wu_k$ are the top edges of $S(w)$, $k\geq 2$, and let $x_1,\ldots ,x_k$ be the corresponding crossing points on $OZ$. Consider two of those edges $wu_i,wu_j$, with $x_i$ placed before $x_j$. Then
		
		\begin{enumerate}[label={\textbf{(b$\mathbf{_\arabic*}$)}}, leftmargin=*]
			\item\label{lem:b1} If $e=u_iu_j$ is bottom, it must cross $c$ at a point $x$ placed on $Ox_i$. If it is top, it must cross $c$ at a point $y$ placed on $x_jZ$. See Figure~\ref{fig:Propertyb1New}.
			
			\item\label{lem:b2} Suppose that there is an edge $e$ of $D_n$ crossing both arc $wx_1$ and arc $wx_k$. Then the two endpoints of $e$, vertices $v,v'$, have to be endpoints of bottom edges of $S(w)$, and not both vertices can be inside the triangle~$x_1x_kw$.
		\end{enumerate}
	\end{lemma}
	\begin{proof}
		\ref{lem:b1} Suppose $e=u_iu_j$ is a bottom edge; the other case follows analogously. Since the two endpoints $u_i,u_j$ are outside the triangular region~$x_ix_jw$, the edge $u_iu_j$ cannot enter in that region without breaking the simplicity, and therefore $u_iu_j$ cannot cross the arc $x_ix_j$. If the edge $e$ crosses $OZ$ at a point $x$ on~$x_jZ$, just after crossing $OZ$, $e$ is outside the region $u_ixx_jw$, but vertex $u_j$ is inside that region, so it is impossible to reach $u_j$ without either breaking the simplicity of $D_n$ or crossing the curve $OZ$ twice. Hence, $e$ must cross $OZ$ on the arc~$Ox_i$. See Figure~\ref{fig:Propertyb1New} left.
		
		\begin{figure}[!htb]
			\centering
			\includegraphics[scale=0.7,page=3]{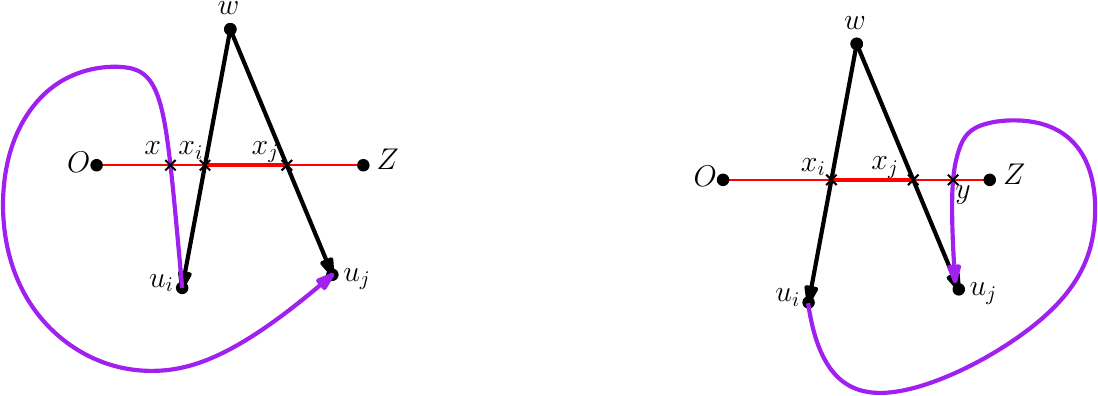}
			\caption{Property $b_1$: The edge $u_iu_j$ must cross $OZ$ as in the left figure or as in the right figure.}\label{fig:Propertyb1New}
		\end{figure}
		
		\ref{lem:b2} Suppose that the edge $e=vv'$ crosses both $wx_1$ and $wx_k$.
				
		\begin{figure}[!htb]
			\centering
			\includegraphics[scale=0.67,page=4]{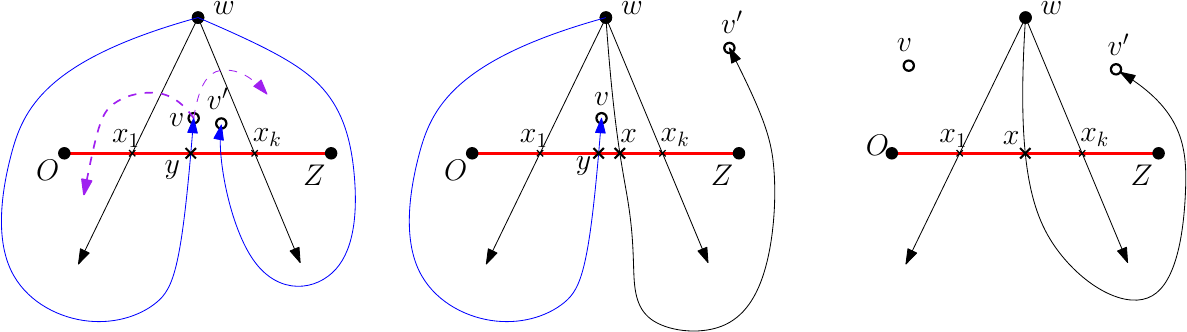}
			\caption{Property $b_2$: Edge $e$ cannot cross both $wx_1,wx_k$ if its endpoints are inside $x_1x_kw$ or when one of them is an endpoint of a top edge.}\label{fig:Propertyb2New}
		\end{figure}

		We first analyze the case when both endpoints of $e$, vertices $v,v'$, are inside the triangular region~$x_1x_kw$, and therefore, by Property~\ref{lem:a2}, both are endpoints of bottom edges of $S(w)$. Let $y$ be the crossing point of $wv$ with $OZ$. Then, $e=vv'$ cannot cross the arc $x_1x_k$ because otherwise the boundary of $x_1x_kw$ is crossed three times, contradicting that both $v,v'$ are in that region. So, $vv'$ has to first cross either the arc $wx_1$ or $wx_k$. In the first case, after crossing $wx_1$, the edge is in the region $wyx_1$ and the vertex $v'$ is outside that region, hence the edge cannot leave that region keeping the simplicity. See Figure~\ref{fig:Propertyb2New} left. Similarly, if $vv'$ first crosses $wx_k$, it enters in the region $wx_ky$, but the vertex $v'$ is outside that region, therefore the edge cannot leave that region without breaking the simplicity.
		
		Suppose now that $v$ is inside the region~$x_1x_kw$, and $v'$ is an endpoint of a top edge $wv'$ crossing $OZ$ at a point $x$. As above, let $y$ be the crossing point of $wv$ with $OZ$, and suppose that $y$ is placed before $x$ on $OZ$. See Figure~\ref{fig:Propertyb2New} centre. Then,  $v$ is in the region $wyx$ and  $v'$ is outside that region, therefore the edge $e=vv'$ has to cross the arc $yx$. On the other hand, since neither $v'$ nor $v$ are in the region $xx_kw$ and $e$ crosses $wx_k$, it has to cross also the arc $xx_k$, hence $e$ should cross $OZ$ twice, a contradiction. A similar analysis can be done in the symmetric case, when $y$ is placed after $x$.
		
		Finally, suppose that $v$ is outside the region~$x_1x_kw$, and $v'$ is an endpoint of a top edge $wv'$ crossing $OZ$ at point $x$. As both vertices $v,v'$ are outside the region $x_1x_kw$, the edge $vv'$ cannot cross three times the boundary of that region, so $e$ cannot cross the arc $x_1x_k$. However, when $e$ enters in that region, by crossing $wx_1$ or $wx_k$, as it cannot cross the arc $wx$, it should cross $x_1x_k$, again a contradiction. See Figure~\ref{fig:Propertyb2New} right.
	\end{proof}

	If we consider a mirror drawing of $D_n$ on the horizontal line $OZ$, all the top edges become bottom and vice versa, then, $\ref{lem:b1}$ has a symmetric Property $\mathbf{(b'_1)}$:
	
	If $wu_i,wu_j$ are bottom edges and the crossing point $x_i$ is placed before $x_j$, then
	if $e=u_iu_j$ is top, it must cross $c$ at a point $x$ placed on $Ox_i$, and  if it is bottom, it must cross $c$ at a point $y$ placed on $x_jZ$.

		\begin{lemma}\label{prop:alltop}
		There is one vertex $w_1$ such that all the edges emanating from $w_1$ are top.
	\end{lemma}
	Lemma~\ref{prop:alltop} is depicted in Figure~\ref{fig:AllTop}.

		\begin{figure}[!htb]
		\centering
		\includegraphics[scale=0.7,page=5]{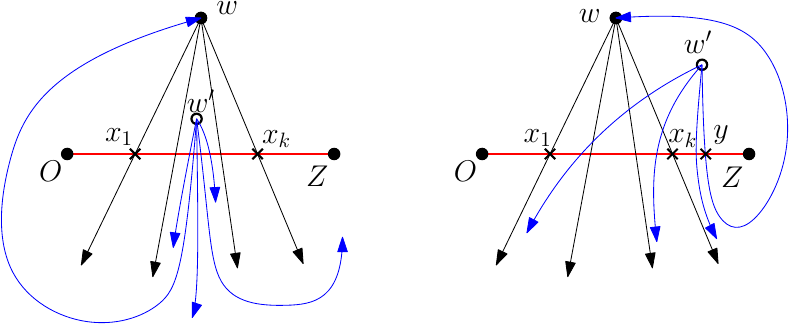}
		\caption{Lemma~\ref{prop:alltop}: There is a vertex $w_1$ such that all the edges of $S(w_1)$ are top.}\label{fig:AllTop}
	\end{figure}

	\begin{proof} We prove the existence of such a vertex by induction on the number of vertices~$n$. For $n=2$ the lemma obviously is true. Now, consider a simple drawing $D_n$ of~$K_n$ and assume that the lemma holds for any simple drawing of~$K_{n-1}$ (with all their edges crossed by a curve). By removing a vertex $w'$ of~$D_n$, we obtain a simple drawing of $K_{n-1}$, by induction, containing a vertex $w$ such that all the edges of $S(w)$ are top. If $ww'$ is also top in $D_n$, all the edges of $S(w)$ are top in~$D_n$, and $w$ is the sought vertex. Suppose now that $ww'$ is bottom and let $x_1,x_k$ be the crossing points of the first and last edges of $S(w)$. If $ww'$ crosses the arc $x_1x_k$,  $w'$ has to be inside the region ~$x_1x_kw$. Then, as all the other vertices are outside that region,  by Property$\ref{lem:b2}$, all the edges of $S(w')$ must cross first $x_1x_k$ and therefore they are top. Finally, if $ww'$ is bottom and does not cross $x_1x_k$ it has to cross $OZ$ through a point $y$ on the arc $x_kZ$ or on the arc $Ox_1$. Suppose $y$ is on $x_kZ$ (the other case is symmetric), then $w'$ is in the triangular region~$x_1yw$, and all the other vertices out of that region. Thus, again by Property$\ref{lem:b2}$,  all edges incident to $w'$ have to cross first the top part of~$OZ$, and therefore $w'$ is the sought vertex.
	\end{proof}

	By removing~$w_1$, we obtain a subdrawing with one vertex $w_2$ of which all the edges are top, by removing $w_2$ another vertex~$w_3$, and so on. Thus, we obtain an order of the vertices $w_1,\ldots ,w_n$ such that for each vertex $w_i$ \textbf{the edges $w_iw_j$ with $j>i$ are top edges} ($n-i$ edges) \textbf{and the edges~$w_iw_l$, with $l<i$ are bottom edges} ($i-1$ edges).

	\subsection*{Construction of the drawing $D_{n+2}$}
	
	Beginning with the simple drawing $D'_0$ formed by $D_n$ and the curve $c=OZ$, we build new drawings $D'_1,\ldots ,D'_i,\ldots ,D'_n$, where drawing $D'_i$ is obtained from drawing $D'_{i-1}$ by adding two simple curves $Ow_i,w_iZ$, following the order $w_1,w_2,\ldots w_n$ explained above. Hence, we start adding $Ow_1,w_1Z$, where $w_1$ is the unique vertex such that all the edges of $S(w_1)$ are top, and we finish by adding $Ow_n,w_nZ$, where $w_n$ is the unique vertex such that all the edges of $S(w_n)$ are bottom.  We build the new drawings in such a way that in $D'_i$, the following invariants are satisfied:
	
	\begin{enumerate}
		\item\label{constr1} The curves $OZ,Zw_i,w_iO$ form a well defined triangular region (region $R_i=OZw_i$) containing the triangular region $R_{i-1}=OZw_{i-1}$, and containing no vertex $w_j$ with~$j>i$. Notice that this invariant implies that $R_i$ contains precisely the vertices $w_l$ with $l<i$, and that neither $Ow_i$ nor $Zw_i$ can properly cross any edge $Ow_l$ or $Zw_l$ with $l<i$.
		\item\label{constr2} The drawing $D'_i$ is a simple drawing.
	\end{enumerate}
	
	The last drawing obtained, $D'_n$, taking $O$ and $Z$ as vertices, provides the sought drawing $D_{n+2}$, since Invariants~\ref{constr1} and~\ref{constr2} imply that
it satisfies properties~\ref{p:simple} and~\ref{p:OZplane}.
	
	To prove that $D'_i$ satisfies Invariants~\ref{constr1} and ~\ref{constr2}, we  suppose that $D'_{i-1}$ satisfies these two invariants. Then, the crossing points of the edges of $D'_{i-1}$ with the boundary of $OZw_{i-1}$ must satisfy the properties given in the following lemma.
	\begin{lemma}\label{prop:propertiesc}
		Suppose that $D'_{i-1}$ satisfies Invariants~\ref{constr1}, ~\ref{constr2}, and let $x_s,x_t,x$ be the crossing points of $OZ$ with the edges $w_{i-1}w_s,w_{i-1}w_t,w_sw_t$, respectively. Then
		
		\begin{enumerate}[label={\textbf{(c$\mathbf{_\arabic*}$)}}, leftmargin=*]
			\item\label{lem:c1} All the top edges of $S(w_{i-1})$ are counterclockwise between $w_{i-1}O$ and~$w_{i-1}Z$, and the bottom edges are counterclockwise between $w_{i-1}Z$ and~$w_{i-1}O$.
			
			\item\label{lem:c2} Any edge $w_sw_t$ with $s<t<i-1$ crosses first $OZ$ then one of the curves $Ow_{i-1}$ or $Zw_{i-1}$.
			
			Besides, if $w_sw_t$ crosses $Zw_{i-1}$, the order of the above crossing points on $OZ$ is $x_t,x_s,x$, and if $w_sw_t$ crosses $Ow_{i-1}$, this order is $x,x_s,x_t$.
			
			\item\label{lem:c3} Any edge $w_sw_t$ with $s<i-1<t$ crosses first $OZ$ and it does not cross any other arc of the boundary of~$OZw_{i}$.
			
			Besides, if $w_s$ is in the region $Ox_tw_{i-1}$, the order of the above crossing points on $OZ$ is $x_s,x,x_t$, and if $w_s$ is in the region $x_tZw_{i-1}$, the order is $x_t,x,x_s$.

			\item\label{lem:c4} Any edge $w_sw_t$ with $i-1<s<t$ first crosses one of $Ow_{i-1}$ or $Zw_{i-1}$, then~$OZ$. Besides, if $w_sw_t$ crosses $Zw_{i-1}$, the order of the crossing points is $x_s,x_t,x$, and if it crosses $Ow_{i-1}$, the order is $x,x_t,x_s$.
			
		\end{enumerate}
	\end{lemma}
	\begin{proof}
		\ref{lem:c1} A top edge $w_{i-1}w_t$ with $t>i-1$ of $S(w_{i-1})$ has to reach $OZ$ by inside of $OZw_{i-1}$, and since $D'_{i-1}$ is simple, it must start entering in that region, hence it must be counterclockwise between $w_{i-1}O$ and~$w_{i-1}Z$. The same reasoning works for bottom edges $w_{i-1}w_s$ with $s<i-1$.
		
		\ref{lem:c2} Since $s<t<i-1$, by Invariant~\ref{constr1}, the vertices $w_s,w_t$ are  both in $OZw_{i-1}$. Hence, the edge $w_sw_t$ must cross the boundary of $OZw_{i-1}$ an even number of times. And since it has to cross $OZ$, it has to cross also only one of the boundary curves $Ow_{i-1},w_{i-1}Z$. Finally, since $w_sw_t$ is a top edge, it has to cross first $OZ$ then the other boundary curve.
		
		On the other hand, suppose that the edge $w_sw_t$ crosses $w_{i-1}Z$ at a point $z$. Then, the vertex $w_{i-1}$ is in the region $xZz$, like the vertices $w_s,w_t$. Therefore, if $w_tw_{i-1}$ (or $w_sw_{i-1}$) leaves that region crossing the arc $xZ$, it cannot reach $w_{i-1}$ without crossing again the boundary of that region, breaking the simplicity of $D'_{i-1}$. Therefore, $x_s,x_t$ have to be placed before $x$ on $OZ$. Finally, $w_tw_{i-1}$ cannot cross the arc $x_sw_s$ (on the edge $w_{i-1}w_s$) or the arc $w_sx$ (on the edge $w_sw_t$), therefore it cannot cross $OZ$ between $x_s$ and $x$, and the order of the crossing points must be $x_t,x_s,x$. The same arguments can be used when the edge $w_sw_t$ crosses $Ow_{i-1}$
		
		\ref{lem:c3} By Invariant~\ref{constr1}, $w_s$ is inside the region $OZw_{i-1}$ and $w_t$ is outside, so the edge $w_sw_t$ has to cross the boundary of $OZw_{i-1}$ an odd number of times. Suppose that it crosses the three curves of the boundary, curves $OZ,Ow_{i-1},w_{i-1}Z$. It cannot cross first $Ow_{i-1}$ (or $w_{i-1}Z$) then $OZ$ because $w_sw_t$ is a top edge. It cannot cross first $OZ$ then $Ow_{i-1}$ and finally $w_{i-1}Z$ because then, by ~\ref{lem:c1},  it has to cross the top edge $w_{i-1}w_t$. By the same reason it cannot cross the boundary in the order first $OZ$, next $w_{i-1}Z$ and then $Ow_{i-1}$. Finally, if it crosses first $Ow_{i-1}$, next $w_{i-1}Z$ and then $OZ$, then, again  by ~\ref{lem:c1}, it has to cross  the bottom edge $w_{i-1}w_s$.
		
		Besides, since $w_{i-1}w_t$ is a top edge, the arc $w_{i-1}x_t$ divides the region $OZw_{i-1}$ into two disjoint regions $Ox_tw_{i-1}$ and $x_tZw_{i-1}$.
		Then, if $w_s$ is in the region $Ox_tw_{i-1}$, the edge $w_sw{i-1}$ has to cross by the arc $Ox_t$. And then, necessarily the crossing point $x$ must be between $x_s$ and $x_t$. The same argument can be used when $w_s$ is in the region $x_tZw_{i-1}$.
		\ref{lem:c4} Since $i-1<s<t$, both vertices $w_s,w_t$ are outside region $OZw_{i-1}$, and we use the same reasonings as in  ~\ref{lem:c2}: The boundary of $OZw_{i-1}$ has to be crossed twice, and since $w_sw_t$ is a top edge, first one of $Ow_{i-1}$ or $Zw_{i-1}$ is crossed, then $OZ$.

		Now, if $w_sw_t$ crosses $w_{i-1}Z$ at a point $z$, then none of the edges $w_{i-1}w_s,w_{i-1}w_t$ can cross the arc $xz$ (placed on $w_sw_t$), so they cannot cross the arc $xZ$ either. Therefore the crossing points $x_s,x_t$ have to be placed before $x$ on $OZ$. Finally, $w_{i-1}$ and $w_t$ are in different sides of the region $x_sxw_s$ and the edge $w_{i-1}w_t$ cannot cross the arc $w_sx_s$ (on edge $w_{i-1}w_s$) or the arc $w_sx$ (on the edge $w_sw_t$), therefore it has to cross $OZ$ between $x_s$ and $x$, and the order of the crossing points must be $x_s,x_t,x$. The same arguments can be used when the edge $w_sw_t$ crosses $Ow_{i-1}$.
	\end{proof}

	\begin{figure}[!htb]
		\centering
		\includegraphics[scale=0.5,page=6]{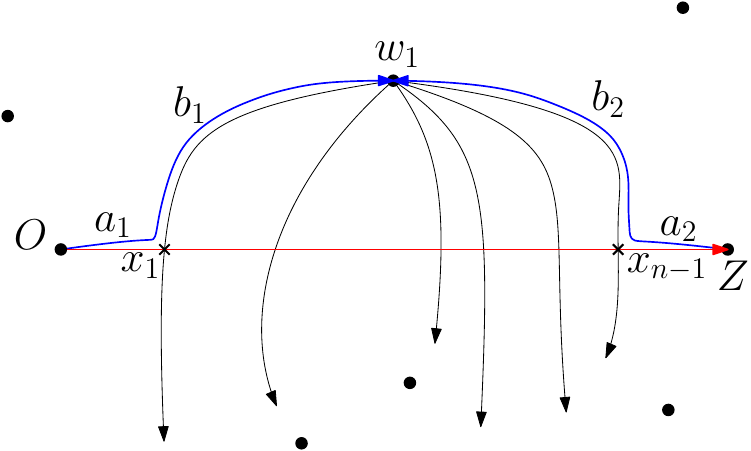}
		\caption{Step 1: How to draw $Ow_1,Zw_1$.}\label{fig:Step1}
	\end{figure}

	Observe that, as we suppose that $D'_{i-1}$ satisfies Invariants~\ref{constr1} and ~\ref{constr2}, properties~\ref{lem:c2}, ~\ref{lem:c3}, ~\ref{lem:c4} imply that no edge of $D_n$ can cross simultaneously the curve $Ow_{i-1}$ and the curve $w_{i-1}Z$.

	For the construction of the drawings $D'_{i}$, we will make a first step (for $D'_1$), a generic step (for $D'_i$ with $1<i<n$), and a final step (for $D'_n$).

	\textbf{Step 1.} Let $x_1,\ldots ,x_{n-1}$ denote the crossing points of $S(w_1)$ with $OZ$. We draw $Ow_1$ following (slightly counterclockwise) the curve $OZ$  until the point~$x_1$, then we turn counterclockwise following the arc~$x_1w_1$; see Figure~\ref{fig:Step1}. We draw $Zw_1$ analogously, following $ZO$ until $x_{n-1}$, then following the arc $x_{n-1}w_1$. If we denote the arcs $Ox_1,x_1w_1,Zx_{n-1},x_{n-1}w_1$ by $a_1,b_1,a_2,b_2$ respectively, then $Ow_1$ has a first part glued to $a_1$, then a second part glued to $b_1$, and in the same way, $Zw_1$ consists of two parts glued to $a_2$ and $b_2$ respectively. Clearly, $OZw_1$ is a well defined region and by Property~\ref{lem:a2}, there are no vertices in the triangular region $w_1x_1x_{n-1}$, hence region $OZw_1$ does not contain any vertices. That establishes Invariant~\ref{constr1}.
	
	Suppose now that the edge $w_sw_t,1<s<t\leq n-1$ crosses the arc $a_1$, and let $x_s,x_t$ be the crossings points on $OZ$ of $w_1w_s,w_1w_t$, respectively. Notice that $x_s$ can be placed before or after $x_t$, but $x$ is placed before these two points.
	Then, by Property~\ref{lem:b1}, since the bottom edge is $w_tw_s$, $x_t$ has to be placed before $x_s$, and the order of the crossing points on $OZ$ must be $x,x_1,x_t,x_s$. Besides, the arc $b_1$  must start at $w_1$ by outside the triangular region $w_1w_tw_s$ (because it is the first top arc), and it finishes at  $x_1$, a point also placed  outside  $w_1w_tw_s$. Therefore, the arc $b_1$ cannot cross $w_sw_t$, the unique edge of the triangle $w_1w_tw_s$ not incident to $w_1$. A symmetric argument can be used to prove that an edge $w_sw_t, 1\leq s<t<n-1$ crossing $a_2$ cannot cross $b_2$. This establishes Invariant~\ref{constr2} for $D'_1$.

	\textbf{Step i} ($2\leq i\leq n-1$). We will draw $Ow_i$ in two different ways, Way~1 and Way~2, depending on whether the first top edge of~$S(w_i)$, edge~$e_1$, first crosses $Ow_{i-1}$ or~$w_{i-1}Z$. When $e_1$ first crosses~$Ow_{i-1}$ at a point $z'_1$, we draw~$Ow_i$ in \textbf{Way~1}, which is the following: $Ow_i$ follows the curve $Ow_{i-1}$ (very close to that curve, slightly counterclockwise) until it reaches the crossing point $z'_1$, then $Ow_i$ continues close to $z'_1w_i$ until it reaches its endpoint~$w_i$; see Figure~\ref{fig:Cases1and2} left.

	When $e_1$ crosses first~$w_{i-1}Z$, by Property~\ref{lem:c4} applied to $e_1$ ($e_1$ is some $w_iw_t$ with $i-1<i<t$), the edge $w_{i-1}w_{i}$ has to cross $OZ$ before $x'_1$ (the crossing point of $e_1$ with $OZ$), and therefore, the last bottom edge of~$S(w_i)$, edge ~$e'_{k'}$, crosses $OZ$ at a point $y'_{k'}$ placed before~$x'_1$ (or $e'_{k'}$
	coincides with $w_iw_{i-1}$).
	In this case, we draw~$Ow_i$ in \textbf{Way~2}, which is the following: $Ow_i$ follows the curve~$OZ$ (very close to that curve, slightly clockwise) until it reaches the crossing point~$y'_{k'}$, then $Ow_i$ continues close to the arc $y'_{k'}w_i$ on $e'_{k'}$ until it reaches the endpoint~$w_i$; see Figure~\ref{fig:Cases1and2} right.
	
	\begin{figure}[!htb]
		\centering
		\includegraphics[scale=0.5,page=7]{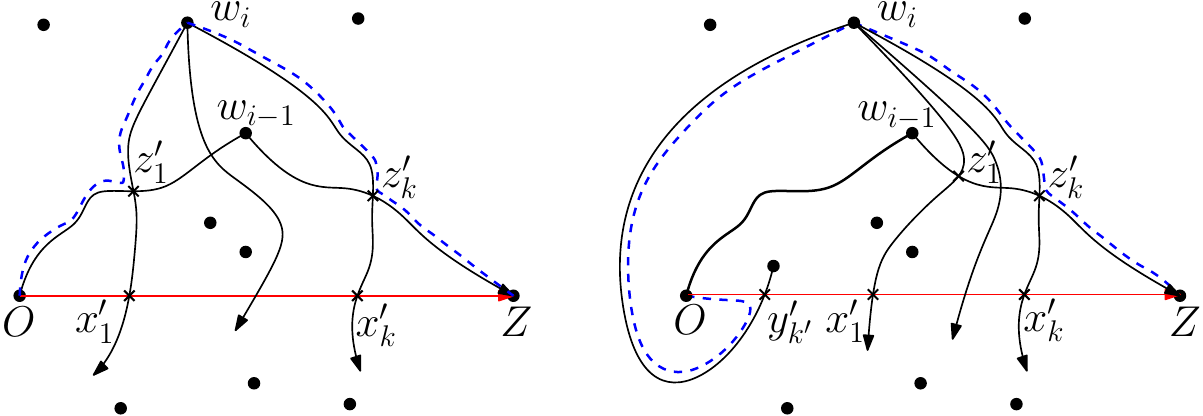}
		\caption{Case 1 and Case 2.}\label{fig:Cases1and2}
	\end{figure}

	In both ways, we consider the curve $Ow_i$ as consisting of two arcs  $a_1$ and $b_1$. In Way~1, the arcs of $Ow_i$ are first $a_1$ glued to $Oz'_1$ and second $b_1$ glued to~$z'_1w_i$. In Way~2, the arcs are first $a_1$ glued to $Oy'_{k'}$ and then $b_1$ glued to~$y'_{k'}w_i$.
	
	Symmetric constructions are used for drawing~$Zw_i$. When $e_k$, the last top edge of $S(w_i)$,  first crosses $w_{i-1}Z$ at a point~$z'_k$, $Zw_i$ is drawn in Way~1: $Zw_i$ consists of the arcs $a_2$ and $b_2$ glued to $Zz'_k$ and $z'_kw_i$ respectively. When $e_k$ first crosses $Ow_{i-1}$, then $OZ$ at a point $x'_k$,  the first bottom edge of~$S(w_i)$, edge $e'_1$, has to cross $OZ$ at a point $y'_{1}$ placed after~$x'_k$. Then, we build $Zw_i$ in Way~2: it consists of an arc $a_2$ (counterclockwise) close to~$Zy'_1$, then an arc $b_2$ glued to~$y'_1w_i$ on edge $e'_1$.
	
	As construction Way~2 can only be used in one of the two curves $Ow_i,Zw_i$, we only need to see that Invariants~\ref{constr1}, \ref{constr2} hold in two cases: \textbf{Case~1}, when construction Way~1 is used for both $Ow_i$ and~$w_iZ$, and \textbf{Case~2}, when Way~2 is used for $Ow_i$ and Way~1 for~$w_iZ$. The case when Way~2 is used for $Zw_{i}$ and Way~1 for $Ow_i$ is symmetric to Case~2; see Figure~\ref{fig:Cases1and2}.
	
	To prove that Invariants~\ref{constr2} holds for~$D'_i$, we will see that in both Case~1 and Case~2,  each edge of $D_n$ crosses at most one of the arcs $a_1,b_1$ and at most one of the arcs $a_2,b_2$ of~$D'_{i}$.

	\textbf{Case~1.-}
	
	By construction, the curves $Ow_i$ and $w_iZ$ do not cross each other, the triangular region $OZw_i$ contains the region~$OZw_{i-1}$, and $Ow_i$, $w_iZ$ cannot properly cross any edge $Ow_l$, $w_lZ$, with $l<i$, because these last edges are inside $OZw_{i}$. Notice that the order of the edges around $O$ and $Z$ are counterclockwise $OZ,Ow_1,\ldots ,Ow_i$ and $ZO,Zw_i,\ldots ,Zw_1$, respectively. Moreover, by Property~\ref{lem:a2}, all the vertices in the triangular region $x'_1x'_{k}w_i$ must be reached from $w_i$ via a bottom edge (where $x'_1,\ldots ,x'_k$ are the crossing points of the top edges of $S(w_i)$ with curve~$OZ$). On the other hand, all the vertices $w_s, s\leq i-1$ are in $OZw_{i-1}$ ($w_{i-1}$ on the boundary), and these are precisely the endpoints of the bottom edges of $S(w_i)$. Therefore, the subregion $z'_1w_{i-1}z'_kw_i$  must be empty, and $OZw_i$ contains all $w_s$ with $s<i$ and does not contain vertices $w_t, t>i$.

	To prove the simplicity of $D'_i$, it is enough to prove that for any edge $w_sw_t$ of $D_n$, the drawing formed by $Ow_i,w_iZ,w_sw_t$ is simple. The following subcases $1,2,3,4,5$ prove that simplicity for edges $w_sw_t$ when $s=i$, $s=i-1$, $s<t<i-1$, $s<i-1<i<t$ and $i<s<t$, respectively. 
	
	\begin{enumerate}[label=\textbf{-\arabic*}]
		\item No edge of $S(w_i)$ can cross any of the arcs~$a_1,b_1,a_2,b_2$.

			\begin{figure}[!htb]
			\centering
			\includegraphics[scale=0.5,page=8]{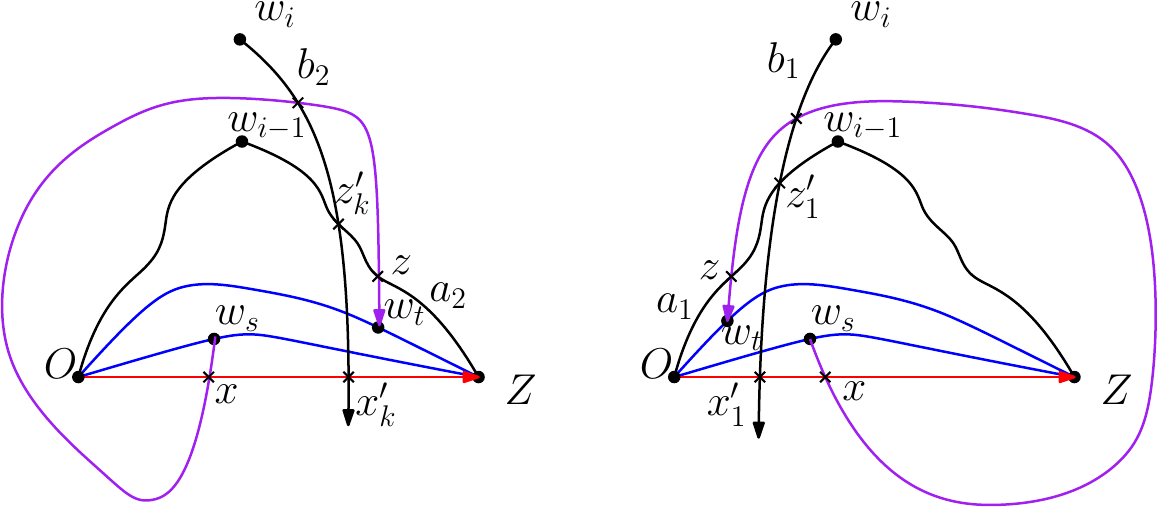}
			\caption{Left: $w_sw_t$ crosses in the order~$OZ,b_2,a_2$. Right: It crosses in the order~$OZ,b_1,a_1$.}\label{fig:CaseiDiff}
		\end{figure}

		\begin{figure}[!htb]
			\centering
			\includegraphics[scale=0.5,page=9]{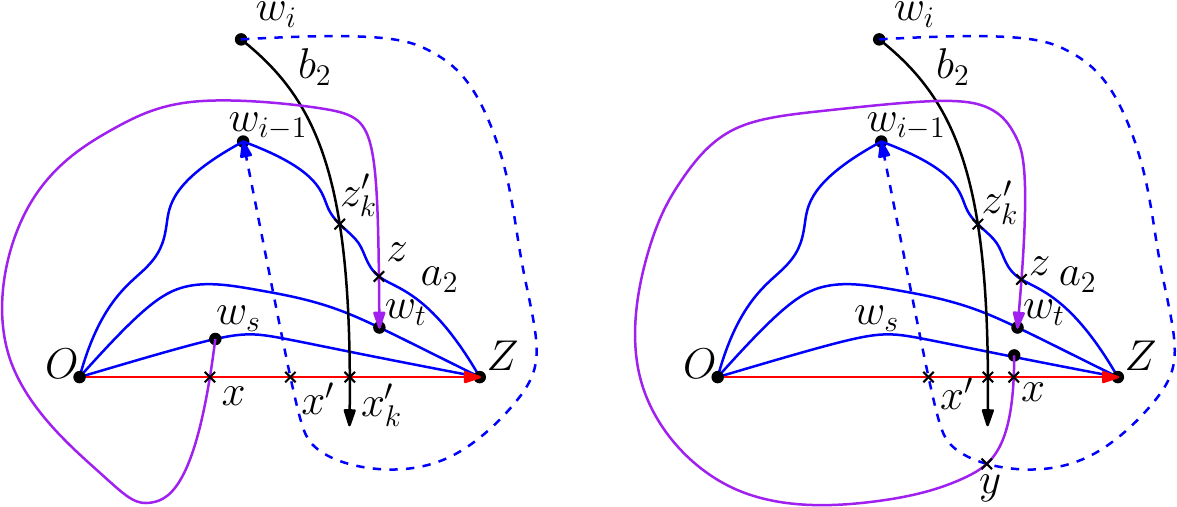}
			\caption{$w_sw_t$ crosses $b_2$ and $a_2$. The edge $w_tw_{i-1}$ cannot be drawn.}\label{fig:CaseiDiff2}
		\end{figure}

		For the arcs $b_1,b_2$, this is obvious since these arcs are glued to edges of $S(w_i)$. Besides, no top edge of $S(w_i)$ can cross arc $a_1$
		because this arc follows the curve $Ow_{i-1}$ until precisely the first crossing point with a top edge of $S(w_i)$. By the same reason, top edges of $S(w_i)$ cannot cross $a_2$ either. Finally, by the simplicity of $D'_{i-1}$ the bottom edge $w_iw_{i-1}$ cannot cross $a_1$ or $a_2$, and by Property~\ref{lem:c3} applied to an edge $w_iw_s$ with $s<i-1<i$, this bottom edge of $S(w_i)$ cannot cross
		$Ow_{i-1}$ or $Zw_{i-1}$, and therefore it cannot cross $a_1$ or~$a_2$.
		
		\enlargethispage{3ex}
		\item Any edge of $S(w_{i-1})$ crosses at most one of the arcs~$a_1,b_1$ and at most one of the arcs $a_2,b_2$.
		
		Since the edges of $S(w_{i-1})$ cannot cross the arc~$a_1$ nor the arc $a_2$, the result follows.

		\item Any edge $e=w_sw_t$, $s<t<i-1$ crosses at most one of the arcs~$a_1,b_1$ and at most one of the arcs $a_2,b_2$.
		
		Suppose that the edge $w_sw_t$ crosses both $a_2$ and $b_2$. By Property~\ref{lem:c2}, this edge first crosses $OZ$ at a point $x$, then reaches $a_2$ at a point $z$ before finishing at $w_t$. Since $b_2=w_iz'_k$ is outside $OZw_{i-1}$ the crossing of $b_2$ with $w_sw_t$ must be at a point in the arc~$xz$, as shown in Figure~\ref{fig:CaseiDiff} left. Observe that after crossing $b_2$, which is an arc on the last top edge of $S(w_i)$, the edge $w_sw_t$ enters in $OZw_{i-1}$ crossing $a_2$, therefore it cannot cross again the last top edge of $S(w_i)$, and thus  $w_t$ has to be placed
		inside the region $x'_kZz'_k$. A totally symmetric case occurs when  the edge $w_sw_t$ crosses both $a_1$ and $b_1$, then $w_t$ must be inside the region $Ox'_1z'_1$, see Figure~\ref{fig:CaseiDiff} right.

		\begin{figure}[!htb]
			\centering
			\includegraphics[scale=0.65,page=10]{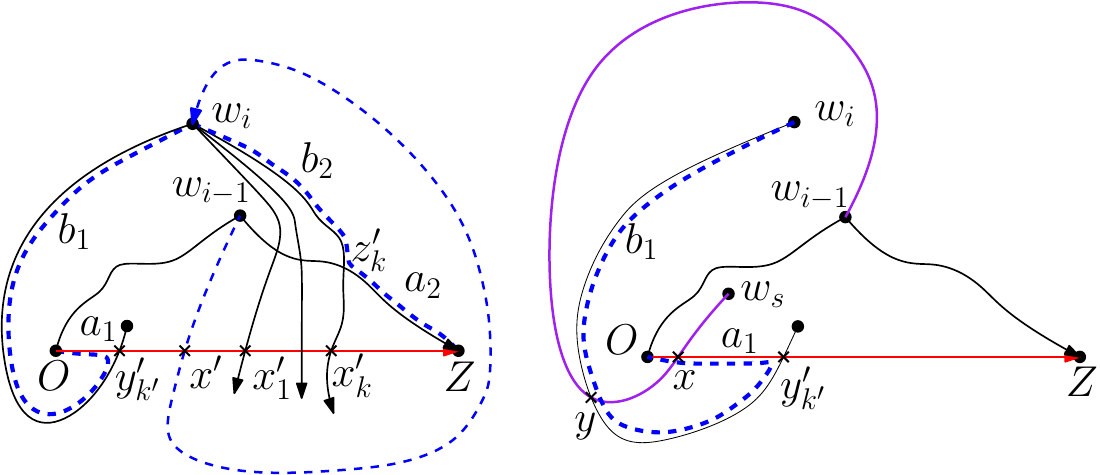}
			\caption{Left: Case 2 for drawing~$Ow_i,Zw_i$. Right: Edges of $S(w_{i-1})$ at most cross one of~$a_1,b_1$.}\label{fig:Casei2Star}
		\end{figure}

		Let us analyze only the first case, when $e=w_sw_t$ crosses~$OZ,b_2,a_2$ in that order, because the other case is totally symmetric. We are going to prove that the edge $w_tw_{i-1}$ cannot be drawn without breaking the simplicity. By Property~\ref{lem:c2} applied to $w_sw_t$, the crossing point of $w_tw_{i-1}$ with $OZ$, must be placed before $x$ on $OZ$. On the other hand, if $x'$ is the crossing point of $w_{i-1}w_{i}$ with $OZ$, by Property~\ref{lem:c4} applied to $e_k$ (the last top edge of $S(w_i)$), the point $x'$ has to be placed before $x'_k$. Then, by Property~\ref{lem:c3} applied to $w_tw_i$ ($t<i-1<i$), as $w_t$ is in region $x'Zw_{i-1}$, the edge $w_tw_{i-1}$ has to cross $OZ$ after~$x'$. Therefore, the edge $w_tw_{i-1}$ should cross $OZ$ after $x'$ and before $x$, and this is not possible when $x$ is placed before $x'$; see Figure~\ref{fig:CaseiDiff2}~(left). Finally, if $x$ is placed after $x'$, then $x'$ is inside the region $xZz$, $w_i$ outside that region, and therefore the arc $x'w_i$ has to cross the arc $xz$ at some point $y$, see Figure~\ref{fig:CaseiDiff2}~(right). But then, the edge $w_tw_{i-1}$ has to cross $OZ$ on the arc $x'x$, entering into the region $x'xy$, bounded by arcs on the edges $w_sw_t$ and $w_iw_{i-1}$, that cannot be crossed by $w_tw_{i-1}$ without breaking the simplicity of~$D_n$.

	\newpage	
		\item Any edge $w_sw_t$, $s<i-1<i<t$ crosses at most one of the arcs~$a_1,b_1$ or $a_2,b_2$.
		
		Observe that $w_s$ is inside the triangular region $OZw_{i-1}$ and
		$w_t$ is outside that region. Then, by Property~\ref{lem:c3}, the edge $w_sw_t$ must first cross~$OZ$ and  it cannot enter the region $OZw_{i-1}$ again. Therefore, it cannot cross $a_1$ or $a_2$.

		\item Any edge $w_sw_t$, $i<s<t$ crosses at most one of the arcs~$a_1,b_1$ or $a_2,b_2$.
		
		First observe that by Property~\ref{lem:b2} applied to $S(w_i)$, the edge $w_sw_t$ cannot cross both $b_1$ and $b_2$ (because $w_iw_s$ and $w_iw_t$ are top edges of that star). Besides, $w_s$ and $w_t$ are outside the triangular region $OZw_i$, therefore if edge $w_sw_t$ crosses $b_1$ or $b_2$ it has to cross also $z'_1w_{i-1}$ or $w_{i-1}z'_k$.
		Finally, according to Property~\ref{lem:c4}, the edge
		$w_sw_t$ must enter to $OZw_{i-1}$ crossing only one of $a_1$,$z'_1w_{i-1}$,$a_2$,$w_{i-1}z'_k$, then leaving through $OZ$, therefore it cannot cross both $a_1,b_1$ or both $a_2,b_2$.
	\end{enumerate}

	\textbf{Case~2.-}

	Again, by the method that the construction is done, the curves $Ow_i$ and $w_iZ$ do not cross each other and the triangular region $OZw_i$ contains the region~$OZw_{i-1}$, hence  $Ow_i$, $w_iZ$ cannot properly cross any edge $Ow_l$, $w_lZ$, with $l<i$.
See Figure~\ref{fig:Cases1and2} right. Moreover, by Property~\ref{lem:a2} all the vertices in region $Ow_{i-1}z'_kw_iy'_{k'}$ are endpoints of bottom edges of $S(w_i)$, but all these endpoints $w_1,\ldots ,w_{i-1}$ of the bottom edges of~$S(w_i)$ must be inside the region $OZw_{i-1}$ ($w_{i-1}$ on the boundary). Therefore, region $Ow_{i-1}z'_kw_iy'_{k'}$ is empty and $OZw_i$ only
	contains inside the $i-1$ endpoints $w_1,\ldots ,w_{i-1}$. So, Invariant~\ref{constr1}, holds.
	
	Like in Case 1, to prove the simplicity of $D'_i$ we analyze the same five subcases.
	
	\begin{enumerate}[label=\textbf{-\arabic*}]
		
		\item\label{Case2_1} No edge of $S(w_i)$ can cross any of the arcs~$a_1,b_1,a_2,b_2$.
		
		As in Case 1, no edge of $S(w_i)$ can cross the arcs~$b_1,b_2$, no a top edge of that star can cross~$a_2$. Besides, by Property~\ref{lem:c3}, a bottom edge $w_iw_s$, $s<i-1<i$, cannot cross~$Zw_{i-1}$, therefore it cannot cross~$a_2$ either. Finally, a top edge cannot cross either arc~$a_1$,
		(because the crossing points $x'_1,\ldots ,x'_k$ are after $y'_{k'}$), and a bottom edge cannot cross $a_1=Oy'_{k'}$ because $y'_{k'}$  is the first crossing point of those edges.

		\item\label{Case2_2} Any edge of $S(w_{i-1})$ crosses at most one of the arcs~$a_1,b_1$ or one of $a_2,b_2$.
		
		See Figure~\ref{fig:Casei2Star} left. Since they cannot cross $a_2$, we only have to prove that one of this edges cannot cross both $a_1$ and $b_1$.
		In the definition of Way 2, we have seen that the crossing point $x'$ of $w_iw_{i-1}$ with $OZ$ must be placed before $x'_1$ and after $y'_{k'}$. Hence,
		if $x_t$ is the  crossing point on $OZ$ of  a top edge $w_iw_t,t>i$, then $x'$ is placed between $y'_{k'}$ and $x_t$. But then, by Property~\ref{lem:c4} applied to $w_iw_t, i-1<i<t$, the crossing point of $w_{i-1}w_t$ on $OZ$ must be placed between $x'$ and $x_t$, so after $y'_{k'}$. Therefore, the top edges $w_{i-1}w_t,t>i-1$ of $S(w_{i-1})$ cannot cross $a_1$.

		Finally, suppose that an edge $w_sw_{i-1},s<i-1$ crosses both $a_1$ and $b_1$. Necessarily, $w_sw_{i-1}$ first crosses $OZ$ at a point $x$, then $b_1$ at a point $y$, finishing at $w_{i-1}$. See Figure~\ref{fig:Casei2Star} right. Then, the vertices $w_i,w_s$ must be both inside the region $xZw_{i-1}$. However, the bottom edge $w_iw_s$ cannot cross $a_1$, (remember that $y'_{k'}$ is the first crossing point of those bottom edges). Therefore, by Property~\ref{lem:c3}, $w_iw_s$ cannot cross $Ow_{i-1}$ or $w_{i-1}Z$, and it has to enter in $OZw_{i-1}$ crossing through the bottom of the arc $y'_{k'}Z$, but then it should first cross the arc $xw_{i-1}$ (on $w_sw_{i-1}$), contradicting the simplicity of $D_n$.
		
		\item\label{Case2_3} Any edge~$e=w_sw_t$, $s<t<i$ crosses at most one of the arcs~$a_1,b_1$ or one of $a_2,b_2$.
		
		A scheme of this situation is shown in Figure~\ref{fig:Casei2Base}, where the two blue curves mark the boundary of regions $OZw_s$ and~$OZw_t$.
		
		\begin{figure}[!htb]
			\centering
			\includegraphics[scale=0.45,page=11]{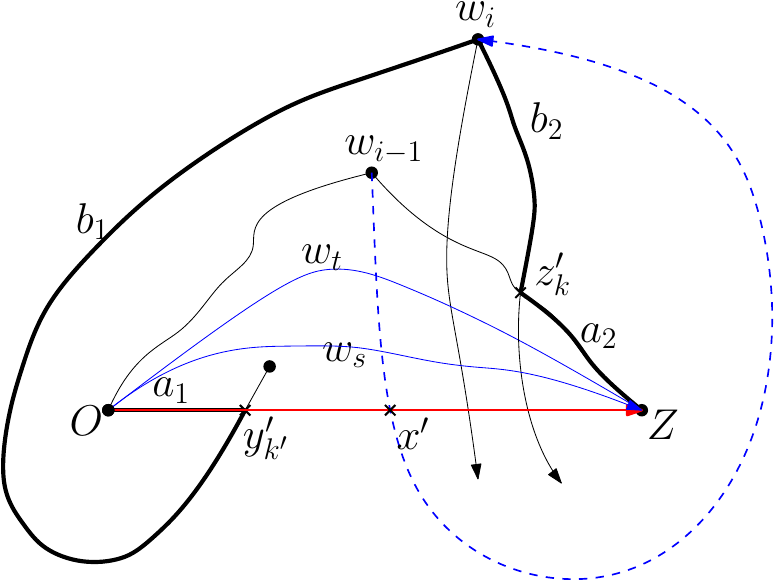}
			\caption{Case~2\ref{Case2_3}.}\label{fig:Casei2Base}
		\end{figure}
		
		On the contrary, suppose that $w_sw_t$ crosses both $a_1$ and $b_1$. Necessarily, it first crosses~$a_1$ at a point $x$, next $b_1$ at a point $y$, then entering in region $OZw_{i-1}$ by crossing either $Ow_{i-1}$ or $w_{i-1}Z$ at a point $z$, until reaching $w_t$. See Figure~\ref{fig:Casei2Basea} left. But then, by the same reasonings as in the previous Case~2\ref{Case2_2}, the edge $w_iw_s$ cannot be drawn. As above, $w_s$ and $w_i$ are in the region $xZz$, if $z$ is on $w_{i-1}Z$,
		or in the region $xZw_{i-1}z$ when $z$ is on $Ow_{i-1}$. However, $w_iw_s$ cannot cross $a_1$ or $Ow_{i-1}$ or $w_{i-1}Z$,  so it has to enter in $OZw_{i-1}$ crossing through the arc $y'_{k'}Z$, and therefore so it should first cross the arc $xz$ (on $w_sw_t$), a contradiction.
		\begin{figure}[!htb]
			\centering
			\includegraphics[scale=0.5,page=12]{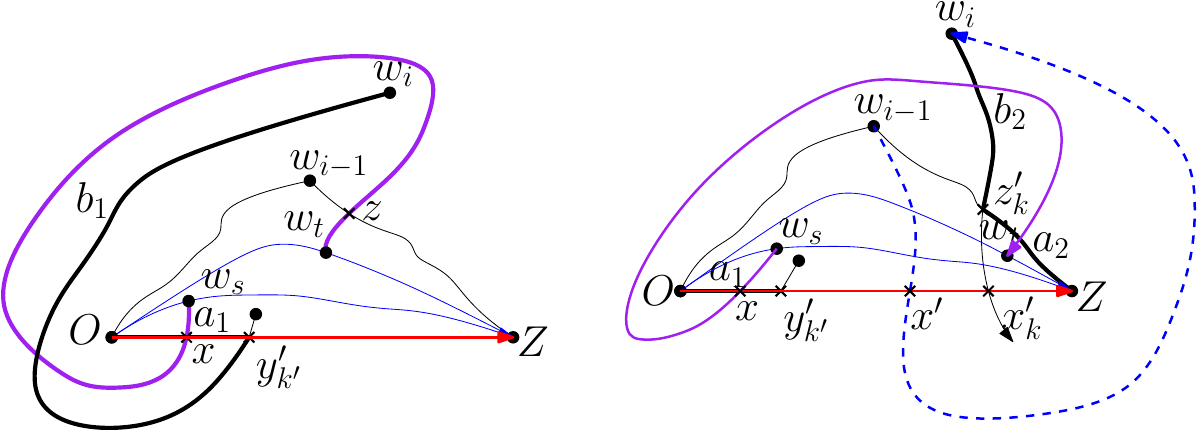}
			\caption{Case~2\ref{Case2_3}: Left, $w_sw_t$ crosses ~$a_1$, then ~$b_1$.
				Right,  $w_sw_t$ crosses $b_2$, then $a_2$.}\label{fig:Casei2Basea}
		\end{figure}

		\begin{figure}[!htb]
			\centering
			\includegraphics[scale=0.4,page=13]{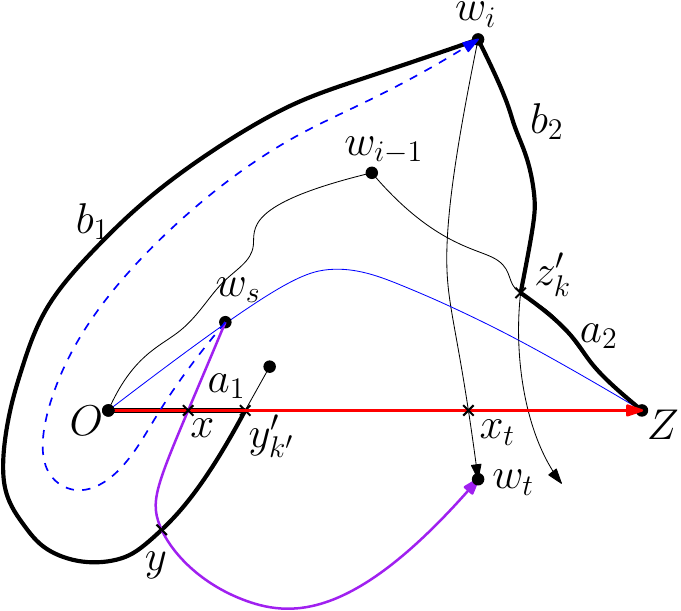}
			\caption{Case~2\ref{Case2_4}.}\label{fig:Casei2Basec}
		\end{figure}

			\enlargethispage{2ex}
		Finally, suppose that $w_sw_t$ crosses $a_2$ and $b_2$. We are exactly in the same situation as in Case~1\ref{Case2_3} (see Figure~\ref{fig:CaseiDiff2}): necessarily the edge first crosses $OZ$ at a point $x$, next crosses $b_2$, then reaches $a_2$ at a point $z$ before finishing at $w_t$. And we have seen that a contradiction is reached in this situation. It does not matter if the edge $w_sw_t$ crosses $a_1$ (like in Figure~\ref{fig:Casei2Basea} right) or not.

		\item\label{Case2_4} Any edge $w_sw_t$, $s<i-1<i<t$ crosses at most one of the arcs~$a_1,b_1$ or $a_2,b_2$.
		
		By Property~\ref{lem:c3}, an edge $w_sw_t$ with $s<i<t$ cannot cross $w_{i-1}Z$, hence, it cannot cross $a_2$. So we have to prove that it cannot cross both $a_1$ and $b_1$. Let us see that in this situation the edge $w_sw_i$ cannot be drawn without breaking the simplicity of the drawing. If $w_sw_t,s<i-1<t$ crosses both $a_1$ and $b_1$, necessarily  it starts crossing $OZ$ through~$a_1$, then crossing $b_1$ at a point $y$, finishing at vertex $w_t$, see Figure~\ref{fig:Casei2Basec}. Then, $w_s$ has to be inside the region $w_iyw_t$. On the other hand, the bottom edge $w_iw_s$ must start from $w_i$ counterclockwise after the arc $b_2$ (on the last top edge of $S(w_i)$), and before the arc $b_1$ (on the last bottom edge of $S(w_i)$). Therefore, $w_iw_s$ should start outside the region $w_iyw_t$, and should finish at $w_s$, placed inside that region. However, $w_iw_s$ cannot cross any of the arcs of the boundary of $w_iyw_t$, because they are on edges incident to either $w_i$ or $w_s$.

		\begin{figure}[htb]
			\centering
			
			\includegraphics[scale=0.4,page=14]{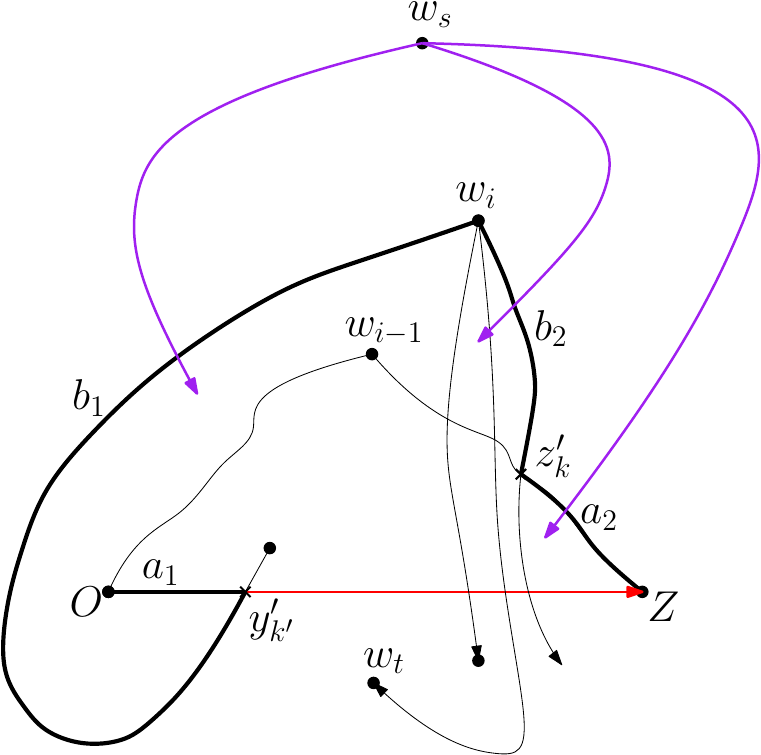}
			
			\caption{Case~2\ref{Case2_5}:$w_sw_t$ enters into region $R_i$  crossing either $b_1$ or $b_2$ or $a_2$.}\label{fig:Casei2Based}
		\end{figure}

		\item\label{Case2_5} Any edge $w_sw_t$,$i<s<t$ crosses at most one of the arcs~$a_1,b_1$ or $a_2,b_2$.

		Since $w_s$ and $w_t$ are outside the triangular region~$OZw_{i-1}$, by Property~\ref{lem:c4}, $w_sw_t$  must cross first either $Ow_{i-1}$ or $w_{i-1}Z$, and then $OZ$. On the other hand, $w_s$ and $w_t$ are outside the region $R=y'_{k'}Zz'_kw_i$, bounded by the four arcs $y'_{k'}Z$,$a_2$,$b_2$,$b_1$, so edge $w_sw_t$ has to enter into that region crossing first either $b_1$ or $b_2$ or $a_2$ and it has to cross either two or four of those arcs.  Suppose that it enters by crossing $b_1$, then it cannot exit by crossing $b_2$ or $a_2$, because then all the top edges of $S(w_i)$ would be crossed, contradicting Property~\ref{lem:b2}. So, it has to exit crossing first $Ow_{i-1}$ or $w_{i-1}Z$, then crossing $y'_{k'}Z$. After crossing $y'_{k'}Z$ it cannot cross again the boundary $Ow_{i-1},w_{i-1}Z$, so it cannot cross $a_2$, and therefore in this case, only $b_1$ and $y'_{k'}Z$ can be crossed. See Figure~\ref{fig:Casei2Based}.

		Similarly, if $w_sw_t$ enters  in $R$ crossing through $b_2$, then it cannot exit by crossing $b_1$, that contradicts the Property~\ref{lem:b2}, not by crossing $a_2$, because then the boundary $Ow_{i-1},w_{i-1}Z$ would be crossed twice, so, it has to exit crossing  $Ow_{i-1}$ or $w_{i-1}Z$, then crossing $y'_{k'}Z$, and again after that crossings, the edge cannot enter again in $R$. Therefore in this case, only $b_2$ and $y'_{k'}Z$ are crossed.
		
		Finally,  if $w_sw_t$ enters  in $R$ by crossing $a_2$, again, it cannot exit by crossing $b_1$, that contradicts the Property~\ref{lem:b2}, not by crossing $b_2$, the boundary $Ow_{i-1},w_{i-1}Z$ would be crossed twice, so, it has to exit crossing  $y'_{k'}Z$. After crossing  $y'_{k'}Z$, the edge cannot cross $b_1$ or $b_2$, because it should cross both, contradicting again Property~\ref{lem:b2}. Therefore in this last case, only $a_2$ and $y'_{k'}Z$ are crossed. 
		
	\end{enumerate}

	These 5 subcases prove that Invariant~\ref{constr2} holds also for~$D'_i$ in Case~2. \\

	\begin{figure}[!htb]
		\centering
		\includegraphics[scale=0.4,page=15]{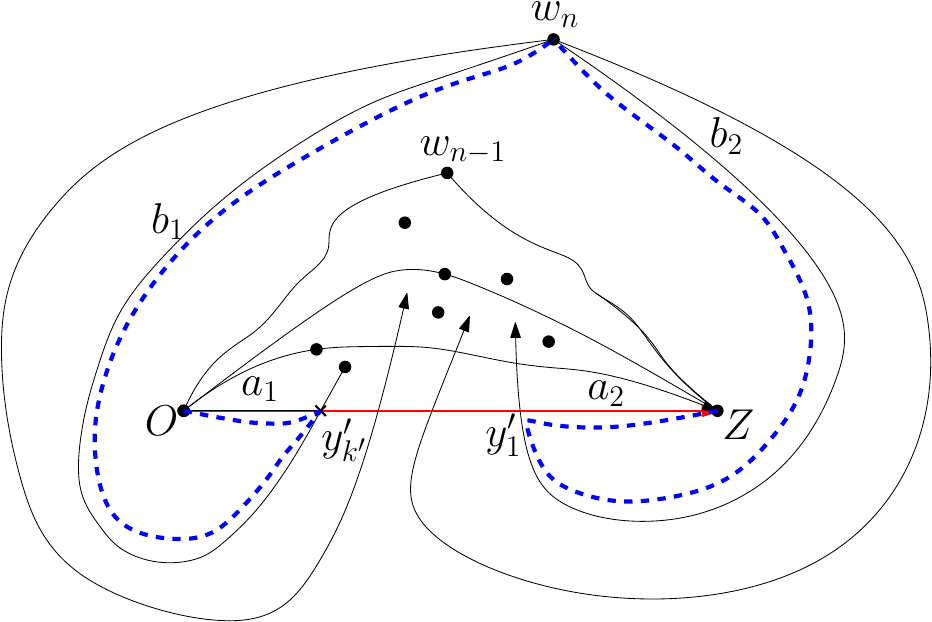}
		\caption{Final Step.}\label{fig:Final}
	\end{figure}

	\textbf{Final Step.-} After $D'_{n-1}$ has been built, we have to add the curves~$Ow_n$, $w_nZ$ to $D'_{n-1}$, where $w_n$ is the last vertex, the one with only bottom edges in~$S(w_n)$. That is done as in \textbf{Step~1}, changing bottom for top, and counterclockwise by clockwise; see Figure~\ref{fig:Final}. Again, by construction, the curves $Ow_n$ and $w_nZ$ do not cross each other, the triangular region $OZw_n$ contains the region~$OZw_{n-1}$, and $Ow_n$, $w_nZ$ cannot properly cross any edge $Ow_l$, $w_lZ$, with $l<n$.
Since the arcs $a_1,b_1$ forming the curve $Ow_n$ are build in \textbf{Way 2}, the reasonings used in Cases~2\ref{Case2_1},2\ref{Case2_2},2\ref{Case2_3}, to prove the simplicity for the arcs $a_1,b_1$, also work in this final step, with $i=n$. By symmetry, the same arguments prove the simplicity for the arcs $a_2,b_2$.
	
\vskip 0.3 cm

	This finishes the proof: The last drawing obtained, $D'_n$, satisfies Invariants ~\ref{constr1},~\ref{constr2}. Then, taking $O$ and $Z$ as vertices, it is the sought drawing $D_{n+2}$, the one satisfying the  Properties~\ref{p:simple} and~\ref{p:OZplane}.
\end{document}